\documentclass[11pt]{article}
\usepackage[utf8]{inputenc}

\usepackage{authblk} 

\usepackage{amsmath,amssymb,amsthm,mathtools,bm}
\usepackage{mathrsfs}
\usepackage{fullpage}
\usepackage{braket}
\usepackage{graphicx,epstopdf,float}
\usepackage{subcaption}
\usepackage[numbers,sort&compress]{natbib}
\usepackage[colorlinks=true,linkcolor=blue,citecolor=green,urlcolor=blue]{hyperref}
\usepackage[nameinlink,capitalise]{cleveref}
\providecommand{\doi}[1]{%
  \href{https://doi.org/#1}{\nolinkurl{#1}}%
}
\usepackage{appendix}

\newtheorem{theorem}{Theorem}[section]

\newtheorem{lemma}[theorem]{Lemma}
\newtheorem{assumption}{Assumption}
\newtheorem{remark}[theorem]{Remark}

\crefname{assumption}{Assumption}{assumptions}

\def\R{\mathbb{R}}
\def\N{\mathbb{N}}
\def\C{\mathbb{C}}
\def\H{\mathcal{H}}
\def\K{\mathcal{K}}
\def\V{\mathcal{V}}

\def\to{\rightarrow}
\def\To{\Longrightarrow}
\def\Ein{E_{\operatorname{in}}}
\def\Eout{E_{\operatorname{out}}}

\newcommand{\norm}[1]{\left\|#1\right\|}
\newcommand{\mat}[1]{\begin{pmatrix}#1 \\ \end{pmatrix}}

\title{\bfseries Sparse Carleman Embeddings for Nonlinear Discrete
Dynamical Systems}
\author[1]{Hsuan-Cheng Wu\thanks{wu.hsuancheng@psu.edu}}
\author[1]{Xiantao Li\thanks{Xiantao.Li@psu.edu}}

\affil[1]{Department of Mathematics, The Pennsylvania State University, University Park, Pennsylvania 16802, USA}

\begin{document}

\maketitle

\begin{abstract}
    Discrete nonlinear dynamical systems arise in population models, iterative optimization, and numerical simulations, where computing many iterations can be costly in high dimensions. We study Carleman embedding for forced quadratic discrete systems and its application to quantum state preparation. Under explicit stability and initial-state assumptions, we prove that the truncation error decreases geometrically with the embedding degree, uniformly over all iterations. A central obstacle in discrete time is that the lifted update can become dense even when the original coefficient maps are sparse. We overcome this obstacle by constructing an exact auxiliary linear-system representation that applies the coefficient maps one tensor factor at a time, with row and column sparsities bounded independently of the truncation degree. After rescaling, the associated history system has a condition number that grows linearly with the iteration count, uniformly in the truncation degree. These results yield a quantum algorithm for preparing an approximation to the normalized $n$th iterate with constant success probability. Under coherent input-access assumptions, and for fixed sparsity, stability parameters, and a known positive lower bound on the final-state norm, the oracle query complexity is linear in $n$ and polylogarithmic in the inverse accuracy. Numerical experiments in one and three dimensions illustrate the convergence of the truncated embedding.
\end{abstract}

\section{Introduction}

Discrete nonlinear dynamics are ubiquitous in scientific computing. They arise whenever a state is updated by a map
$\bm x_{n+1}=F(\bm x_n)$. Such recurrences model populations with nonoverlapping generations \cite{may1976simple} and time-stepping schemes for fluid
dynamics \cite{jennings2026end}. Discrete nonlinear dynamics also arise naturally in machine learning, including gradient-based optimization \cite{liu2024towards,lee2016gradient}, recurrent neural networks \cite{pascanu2013difficulty}, and the fixed-point formulations underlying deep equilibrium models \cite{bai2019deep}. The stability and convergence of these iterations play an important role in the design and analysis of learning algorithms \cite{pascanu2013difficulty,bai2021stabilizing}. These examples are inherently
specified by discrete updates: a discrete dynamical system need not be
a sampling or numerical discretization of any continuous-time evolution.
Even a low-degree map can exhibit behavior
that is difficult to infer from a single update, while a high-dimensional
map can be costly to iterate or probe over many steps. These observations
motivate algorithms that exploit the structure of $F$ to access a later
state or a selected observable. We focus on a stable quadratic regime in which the iterates converge to a fixed point.

A natural route toward quantum algorithms for nonlinear recurrences is to represent their updates as linear evolution in an enlarged space. For polynomial maps, Carleman embedding provides precisely this representation: introduce the monomials of the state, represented here by tensor powers $1,\bm x,\bm x^{\otimes2},\ldots$, and expand their updates under $\bm x_{n+1}=F(\bm x_n)$ to obtain an infinite linear recurrence \cite{berkolaiko1998analysis}. Truncating this hierarchy yields a finite linear approximation to the original nonlinear dynamics. This construction connects nonlinear iteration with quantum algorithms for linear algebra, provided that the approximation error and the cost of accessing the enlarged system can be controlled.

Quantum linear-system algorithms (QLSAs) make this connection computationally useful. Beginning with the algorithm of Harrow, Hassidim, and Lloyd \cite{harrow2009quantum}, these methods prepare quantum states proportional to solutions of linear systems, with subsequent developments improving the dependence on accuracy and conditioning \cite{childs2017quantum,costa2022optimal}. For a truncated Carleman recurrence, the lifted states at successive iterations can be coupled through a single history linear system, following the general strategy used in quantum algorithms for differential equations \cite{Ber14,berry2017quantum,childs2020quantum,krovi2023improved}. Solving this system prepares a state encoding the lifted history, from which a state proportional to the desired nonlinear iterate can be extracted, subject to the truncation error. Under suitable sparsity and coherent-access assumptions, the linear-system solver can have only polylogarithmic dependence on the enlarged dimension.

Realizing this potential requires more than a finite-dimensional approximation. The truncation error must remain controlled over the iteration horizon, the lifted operator must admit efficient quantum access, and the history system must have manageable conditioning. The cost of extracting the desired physical state from the lifted history must also be included. These requirements are coupled: increasing the truncation order enlarges the representation and may increase its implementation cost, even when it improves the approximation. The central question is therefore whether a convergent Carleman approximation can be combined with a quantum implementation whose total cost remains controlled. We address this question for a stable class of quadratic discrete dynamical systems.

Other quantum approaches to nonlinear problems include evolution of
multiple copies of a state through a mean-field approximation
\cite{lloyd2020nonlinear} and Koopman--von Neumann representations of
classical dynamics \cite{joseph2020koopman}. Homotopy perturbation has
also been used to construct quantum algorithms for dissipative nonlinear
ordinary differential equations \cite{xue2021homotopy} and for quadratic
nonlinear algebraic systems \cite{xue2022quadratic}. The latter solves
for a root of a nonlinear equation, whereas our task is to prepare a state
associated with the $n$th iterate of a nonlinear map. These methods
suggest different ways to expose linear structure, with distinct
requirements for approximation, state preparation, and output recovery.

This paper studies Carleman embedding for discrete dynamical systems under specific assumptions. The discrete setting creates the sparsity challenge: The update of each tensor power contains products of the coefficients of
$F$, including terms absent from a continuous-time Carleman generator. Moreover, sparse coefficients need not give a sparse truncated update: for $F(\bm x)=A\bm x$, its degree-$N$ block is
$A^{\otimes N}$, whose row sparsity can grow as $s_A^N$ if $A$ has row
sparsity $s_A$. A convergence bound for the truncation therefore does not
by itself yield an efficient sparse-access quantum implementation. This motivates a representation that applies the original coefficient maps one tensor factor at a time, without expanding their tensor products.

\subsection{Main Contributions}

We study the forced quadratic iteration
\[
    \bm x_{n+1}=\bm f+A\bm x_n+B\bm x_n\otimes \bm x_n.
\]
The Carleman lift collects the tensor powers of the state into
\[
    \bm X_n=(1,\bm x_n,\bm x_n^{\otimes2},\ldots)^T,
    \qquad \bm X_{n+1}=L\bm X_n.
\]
Writing $F_0=\bm f$, $F_1=A$, and $F_2=B$, the blocks of $L$ are
\[
    L_{j,k}
    =\sum_{\substack{\ell_1,\ldots,\ell_j\in\{0,1,2\}\\
                     \ell_1+\cdots+\ell_j=k}}
      F_{\ell_1}\otimes\cdots\otimes F_{\ell_j},
    \qquad j\ge1,\ k\ge0,
\]
with $L_{0,0}=1$ and $L_{0,k}=0$ for $k\ge1$.
Retaining degrees $0,\ldots,N$ defines
\[
    L_N=(L_{j,k})_{j,k=0}^{N},
    \qquad
    \bm y_{n+1}^{(N)}=L_N\bm y_n^{(N)},
    \qquad
    \bm y_0^{(N)}
    =(1,\bm x_0,\ldots,\bm x_0^{\otimes N})^\top.
\]
The degree-zero coordinate remains one, and the degree-one block
$\bm y_{n,1}^{(N)}$ approximates $\bm x_n$.
We establish the following results.
\begin{enumerate}
    \item \textbf{Uniform truncation error.} Under the stability and
    initial-state assumptions in \cref{ass: main}, the degree-one block
    $\bm y_{n,1}^{(N)}$ of the truncated lift satisfies
    \[
        \sup_{n\ge0}\norm{\bm x_n-\bm y_{n,1}^{(N)}}
        \le \rho R^N,\qquad 0\le R<1.
    \]
    For fixed parameters with $0<R<1$, a truncation degree
    $N=O(\log(1/\varepsilon))$ gives error at most $\varepsilon$
    independently of the iteration count.

    \item \textbf{Sparse realization of the lifted update.} We construct
    a nilpotent auxiliary operator $C$ and embeddings $\Ein,\Eout$ for
    which the truncated Carleman matrix has the exact representation
    \[
        L_N=\Ein^\dag D^{-1}\Eout,\qquad D=I-C^\dag.
    \]
    When $\bm f$ has at most $s_f$ nonzeros and the row and column sparsities
    of $A,B$ are bounded by $s_A,s_B$, each row and column of $D$ has at
    most $1+s_f+s_A+s_B$ nonzeros, independent of $N$. We also bound
    $\norm{D^{-1}}$ and its condition number.

    \item \textbf{History state QLSA.} After rescaling, the first $n$ truncated updates form
    a sparse linear system 
    \[ \widetilde{\mathcal A}_n
        =I_n\otimes\widetilde D-S_n\otimes P, 
    \]
    where 
    \[ S_n=\sum_{t=0}^{n-2}\ket{t+1}\bra{t}, \qquad P=\Eout\Ein^\dag.
    \]
    We prove $\kappa(\widetilde{\mathcal A}_n)=\Theta(n)$ with constants
    independent of $N$. Under coherent sparse-access and input-state
    assumptions, the padded history system yields an algorithm for
    preparing a state close to $\bm x_n/\norm{\bm x_n}$. Its cost also
    accounts for output extraction and a lower bound on $\norm{\bm x_n}$.
\end{enumerate}
Combining the uniform truncation bound with the sparse history formulation, we obtain a quantum algorithm for preparing the normalized $n$th iterate. Under coherent input-access assumptions and a known lower bound $0<\xi\le\|\bm x_n\|$, the algorithm achieves trace-distance error at most $\varepsilon$ with constant success probability, using 
\[
O\!\left(snK_r\Gamma\log\frac{\Gamma}{\varepsilon}\right),
\qquad
\Gamma=\frac{2rM_r}{\xi},
\]
coefficient-oracle and right-hand-side preparation queries. Here $s$ bounds the sparsity, while $K_r$ and $M_r$ depend on the stability parameters and rescaling radius. For fixed parameters and $\xi$, the query complexity is linear in the iteration count. Possible improvements of the sparsity dependence from $s$ to $\sqrt{s}$ are discussed in \cref{rem: sparsity_speedup}.

\subsection{Related Work}

Carleman representations of discrete recurrences predate quantum algorithms.
Berkolaiko et al.~\cite{berkolaiko1998analysis} studied polynomial and analytic
recursions, and Myszkowski~\cite{myszkowski2022recursions} developed constructions
for multivariate polynomial recurrences and systems with memory.
Pruekprasert et al.~\cite{pruekprasert2020moment,pruekprasert2024moment}
studied truncated Carleman propagation for discrete-time stochastic polynomial
systems and established approximation-error bounds. Deterministic polynomial
maps are included as a special case of that setting. These results provide
direct precedents for discrete lifting and truncation analysis.

Liu et al.~\cite{liu2024towards} developed Carleman-based quantum algorithms
for discrete gradient updates. Their supplementary analysis treats truncation
error, history linear systems, conditioning, sparsity, and corrections arising
from the discrete product rule. In particular, it includes an error bound
independent of the iteration count under a contraction assumption on the
chosen lifted update. Their main construction uses a small-learning-rate
approximation, while the supplement also derives the higher-order discrete
corrections. Our analysis instead treats the exact forced quadratic update
throughout and derives its guarantees directly from explicit bounds on the
original coefficient maps and initial state.

In fluid dynamics, Sanavio and Succi~\cite{sanavio2024lattice} applied Carleman lifting to the discrete collision and streaming updates of the lattice Boltzmann method, constructed quantum circuits, and numerically investigated the accuracy of low-order truncations. Sanavio et al.~\cite{sanavio2025carleman} subsequently developed matrix-access oracles and block encodings for the discrete Carleman update. They reported sparsity growing exponentially with truncation order and low success probabilities that hinder repeated time stepping. These works provide direct precedents for quantum implementation of discrete Carleman dynamics and motivate controlling the implementation cost as the truncation degree increases.

More recently, Jennings et al.\ developed a discrete Carleman formulation for lattice Boltzmann dynamics and examined time-stepping, conditioning, and data extraction through an end-to-end complexity analysis \cite{jennings2026end}. They also numerically investigated truncation convergence, but did not establish a rigorous convergence guarantee for that application. Our work considers forced quadratic maps under explicit stability assumptions and provides an exact sparse auxiliary realization of the truncated Carleman update, a truncation-error bound uniform in the iteration count, and a history-system analysis with sparsity independent of the tensor degree. The resulting quantum algorithm prepares an encoded solution state. Observable estimation and full classical readout incur additional costs.

\subsection{Outline}
The remainder of the paper is organized as follows. In \cref{sec: math-form}, we formulate the quadratic iteration and its
Carleman embedding, then establish a truncation-error bound that is uniform in the number of iterations. In \cref{sec:sparse-realization}, we construct an auxiliary sparse system that represents the truncated Carleman update exactly, without explicitly forming its potentially dense matrix. In \cref{sec:rescaled-history}, we develop the rescaled history-system formulation and analyze its conditioning, the recovery of the physical state, and the quantum algorithm's query complexity. Finally, \cref{sec: num} presents numerical experiments illustrating the convergence of the truncated embedding. All of the proofs are deferred to appendices.

\paragraph{Notation.}
Throughout, $\|\cdot\|$ denotes the Euclidean norm for vectors and the induced operator norm for matrices. $\|\cdot\|_1$ denotes the trace norm for matrices. For an invertible matrix $M$, we write
$\kappa(M)=\|M\|\|M^{-1}\|$ for its condition number.
The superscripts $T$ and $\dagger$ denote the transpose
and conjugate transpose, respectively. We use $\otimes$ for the standard Kronecker product and adopt
the convention $\bm x^{\otimes 0}=1$. The symbols $I_k$ and $I_{\mathcal H}$ denote the identity
on $\mathbb C^k$ and $\mathcal H$, respectively. $\N_0$ denotes the set of nonnegative integers. All vectors are in \textbf{bold} font.

\section{Mathematical Formulation} \label{sec: math-form}

\subsection{Problem Setup}
We consider discrete quadratic iteration
\begin{equation} \label{eq: main eq}
    \bm x_{n+1} = \bm f + A\bm x_n + B \bm x_n \otimes\bm x_n, \qquad \bm x_0 \in\mathbb{C}^d,
\end{equation}
for any $n\in\N_0$, where $\bm x_n\in \C^d$, $\bm f\in \C^d$, $A: \C^d\to \C^d$ and $B: \C^{d^2}\to\C^d$.

We note that the convergence argument and the method for building the quantum linear system can be extended to higher-degree polynomial discrete systems. We provide further details in \cref{remark: high_orders}. In the remainder of this paper, we focus on the quadratic case for simplicity.

\begin{assumption}\label{ass: main}
There are known bounds $\norm A\le \alpha<1$ and $\norm B\le\beta$,
and the known forcing norm is $c=\norm{\bm f}$, with $\beta>0$,
\begin{equation}\label{eq:roots}
 \Delta=(1-\alpha)^2-4\beta c>0,\qquad
 r_\pm=\frac{1-\alpha\pm\sqrt\Delta}{2\beta},\qquad
 u:=\norm{\bm x_0}<r_+.
\end{equation}
\end{assumption}

Carleman embedding replaces the nonlinear iteration by a linear evolution of all tensor powers of the state. Let $H=\mathbb C^d$ and adopt the convention $\bm x^{\otimes 0}=1$. Define
\begin{equation}
    F_0=\bm f,\qquad F_1=A,\qquad F_2=B,
\end{equation}
where $F_\ell:H^{\otimes\ell}\to H$. For $j\ge1$, set
\begin{equation}\label{eq:carleman-blocks}
    L_{j,k}
    :=
    \sum_{\substack{\ell_1,\ldots,\ell_j\in\{0,1,2\}\\
                    \ell_1+\cdots+\ell_j=k}}
    F_{\ell_1}\otimes\cdots\otimes F_{\ell_j},
    \qquad 0\le k\le 2j,
\end{equation}
and set $L_{j,k}=0$ for $k>2j$. Each block
$L_{j,k}:H^{\otimes k}\to H^{\otimes j}$ collects the terms of input degree $k$ in the expansion
\begin{equation}\label{eq:exact-tensor-update}
    \bm x_{n+1}^{\otimes j}
    =
    \left(\bm f+A\bm x_n+B\bm x_n^{\otimes2}\right)^{\otimes j}
    =
    \sum_{k=0}^{2j}L_{j,k}\bm x_n^{\otimes k}.
\end{equation}
In particular, $L_{1,0}=\bm f$, $L_{1,1}=A$, $L_{1,2}=B$, and
$L_{j,0}=\bm f^{\otimes j}$.

\paragraph{The Carleman embedding system.}
Carleman introduced a method for embedding nonlinear polynomial ODEs into an infinite-dimensional linear system by lifting the state to all monomials \cite{carleman1932application}. Here, we apply the same lifting principle to the discrete-time quadratic iteration, obtaining an infinite-dimensional linear recurrence for the lifted state.
\begin{equation}
    \bm X_n
    :=
    \begin{pmatrix}
        1\\
        \bm x_n\\
        \bm x_n^{\otimes2}\\
        \bm x_n^{\otimes3}\\
        \vdots
    \end{pmatrix}.
\end{equation}
Setting $L_{0,0}=1$ and $L_{0,k}=0$ for $k\ge1$, the exact nonlinear iteration becomes the linear recurrence
\begin{equation}\label{eq:infinite-carleman}
    \bm X_{n+1}=L\bm X_n,
    \qquad
    L= \mat{
        1 & 0 & 0 & 0 & 0 & \cdots\\
        \bm f & A & B & 0 & 0 & \cdots\\
        L_{2,0} & L_{2,1} & L_{2,2} & L_{2,3} & L_{2,4} & \cdots\\
        L_{3,0} & L_{3,1} & L_{3,2} & L_{3,3} & L_{3,4} & \cdots\\
        \vdots & \vdots & \vdots & \vdots & \vdots & \ddots}.
\end{equation}
This identity is interpreted blockwise: each block row contains only finitely many nonzero blocks, so no summability assumption on the full infinite vector is required.

For continuous-time quadratic ODE systems, the Carleman embedding yields a block-tridiagonal matrix \cite{liu2021efficient}. In contrast, the discrete-time Carleman matrix in \eqref{eq:infinite-carleman} generally lacks this structure.

\paragraph{Exact equations for the retained degrees.}

Fix a truncation degree $N\ge1$ and retain
\begin{equation}
    \bm X_n^{(N)}
    :=
    \begin{pmatrix}
        1\\
        \bm x_n\\
        \vdots\\
        \bm x_n^{\otimes N}
    \end{pmatrix},
    \qquad
    L_N:=(L_{j,k})_{0\le j,k\le N}.
\end{equation}
The matrix $L_N$ has dimension
$\sum_{j=0}^{N}d^j$. Splitting
\cref{eq:exact-tensor-update} into retained and omitted input degrees gives the exact identity
\begin{equation}\label{eq:finite-exact-carleman}
    \bm X_{n+1}^{(N)}
    =
    L_N\bm X_n^{(N)}+\bm\tau_n^{(N)},
\end{equation}
where
\begin{equation}\label{eq:carleman-remainder}
    \tau_{n,0}^{(N)}=0,
    \qquad
    \bm\tau_{n,j}^{(N)}
    :=
    \sum_{k=N+1}^{2j}L_{j,k}\bm x_n^{\otimes k},
    \qquad 1\le j\le N.
\end{equation}
An empty sum is zero. Thus $\bm\tau_{n,j}^{(N)}=0$ whenever $2j\le N$.

\cref{eq:finite-exact-carleman} is exact, but it is not a closed linear recurrence for the retained variables: the remainder contains tensor powers of degree greater than $N$.
\cref{eq:finite-exact-carleman} is useful for error analysis, which will be compared with the truncated Carleman system afterwards. 

\paragraph{The truncated Carleman system.}
Dropping the remainder in \cref{eq:finite-exact-carleman} defines the finite Carleman approximation
\begin{equation}\label{eq:finite-truncated-carleman}
        \bm y_{n+1}^{(N)}=L_N\bm y_n^{(N)},
        \qquad
        \bm y_0^{(N)}=\bm X_0^{(N)}.
\end{equation}
The degree-zero coordinate remains exactly one. Consequently, the positive-degree equations are affine:
\begin{equation}\label{eq:finite-carleman-positive}
    \bm y_{n+1,j}^{(N)}
    =
    \bm f^{\otimes j}
    +
    \sum_{k=1}^{\min(N,2j)}
    L_{j,k}\bm y_{n,k}^{(N)},
    \qquad 1\le j\le N.
\end{equation}
The approximation to the original state is the degree-one block
$\bm y_{n,1}^{(N)}$.

To distinguish the omitted terms from the accumulated approximation error, define
\begin{equation}
    \bm \eta_n
    :=
    \bm X_n^{(N)}-\bm y_n^{(N)}.
\end{equation}
Subtracting \eqref{eq:finite-truncated-carleman} from
\eqref{eq:finite-exact-carleman} yields
\begin{equation}\label{eq:carleman-error-recurrence}
    \bm \eta_{n+1}
    =
    L_N\bm \eta_n+\bm\tau_n^{(N)},
    \qquad
    \bm \eta_0=0.
\end{equation}
Since $\bm \eta_{n,0}=0$, its positive-degree blocks satisfy
\begin{equation}
    \bm \eta_{n+1,j}
    =
    \sum_{k=1}^{\min(N,2j)}
    L_{j,k}\bm \eta_{n,k}
    +
    \sum_{k=N+1}^{2j}
    L_{j,k}\bm x_n^{\otimes k}.
\end{equation}
The second sum introduces truncation error only when $2j>N$, while the first sum propagates existing errors between retained degrees. Therefore lower-degree errors, including the error in the physical block $j=1$, can become nonzero even when their equations have no direct omitted-degree term.

\subsection{Convergence of Carleman embedding}
We now state the uniform convergence of the Carleman embedding at finite truncation level-$N$.
\begin{theorem}[Uniform geometric convergence]\label{thm: uniform}
Under \cref{ass: main}, let $\rho=\max\{u,r_-\}$ and $R=\rho/r_+<1$. For every $N\in\N$, $n\ge0$, and
$1\le j\le N$,
\begin{equation}\label{eq:uniform}
 \norm{\bm x_n^{\otimes j}-\bm y^{(N)}_{n,j}}
 \le r_+^jR^{N+1}.
\end{equation}
In particular,
\begin{equation}\label{eq:firsterror}
 \sup_{n\ge0}\norm{\bm x_n-\bm y^{(N)}_{n,1}}\le\rho R^N.
\end{equation}
\end{theorem}
We defer the proof to \cref{pf: thm-uniform}.

\begin{remark} \label{remark: high_orders}
    The formal procedure extends to polynomial maps
    $\bm x_{n+1}=\sum_{\ell=0}^{p}F_\ell\bm x_n^{\otimes\ell}$
    of arbitrary finite degree. Define the scalar majorant
    $h(s)=\sum_{\ell=0}^{p}a_\ell s^\ell$, where
    $a_\ell\ge\norm{F_\ell}$. If there exist $0<\rho<r$ such that $\norm{\bm x_0}\le\rho$, $h(\rho)\le\rho$, and $h(r)\le r$, then the same coefficient-majorization argument yields
    \begin{equation} \label{eq: general_case_bound}
        \sup_{n\ge0}\norm{\bm x_n-\bm y_{n,1}^{(N)}} \le \rho\left(\frac{\rho}{r}\right)^N,
    \end{equation}
    where $\bm y^{(N)}$ is the degree-$N$ Carleman truncation initialized by the tensor powers of $\bm x_0$. Thus, for fixed $\rho$ and $r$, a truncation degree $N=O(\log(1/\varepsilon))$ suffices for error at most $\varepsilon$, uniformly in the number of iterations. When $p=2$, taking $r=r_+$ and $R=\rho/r_+$ recovers \cref{eq:firsterror}.
\end{remark}

From now on, we use $\bm y$ to denote $\bm y^{(N)}$ for simplicity.

\section{Sparse Realization of the Truncated Carleman Embedding} \label{sec:sparse-realization}

Recall the quadratic iteration \cref{eq: main eq},
\begin{equation}
    \bm x_{n+1}=F(\bm x_n)=\bm f+A\bm x_n+B\bm x_n^{\otimes 2},
    \qquad \bm x_n\in \H:=\mathbb C^d.
\end{equation}
We use Euclidean vector norms and their induced operator norms throughout.
Assume that $A$ and $B$ have at most $s_A$ and $s_B$ nonzero entries in every row and column, respectively, and that $\bm f$ has at most $s_f$ nonzero entries.

The sparsity of these coefficients does not imply sparsity of the discrete Carleman matrix. Even if $\bm f=B=0$, the degree-$N$ block is $A^{\otimes N}$, whose row sparsity can be $s_A^N$. Thus explicitly expanding the discrete Carleman matrix can destroy sparsity. We avoid this expansion by introducing auxiliary variables that process one original tensor factor at a time.

\subsection{A Scalar Example at $N=2$}
Before we show that how to handle the sparsity in general, we first show that how to overcome it by a simple example with truncation level $N=2$. Consider a scalar example,
\begin{equation} \label{eq: example}
    x_{n+1}=F(x_n)= f+ax_n+bx_n^2, \quad f,a,b\in\R.
\end{equation}
Expanding $F(x)^2$, one has
\begin{equation}
    F(x)^2 = f^2+ 2afx+ (a^2+2bf)x^2 + 2abx^3 + b^2x^4.
\end{equation}
The Carleman matrix at truncation level $N=2$ is
\begin{equation}
    L_2 = \mat{1 & 0 & 0 \\ f & a & b \\ f^2 & 2af & a^2+2bf}.
\end{equation}
To expand the power of $F(x)$ one factor at a time,  introduce the auxiliary expressions $x^pF(x)^\ell$, where $p$ records the accumulated input degree and $\ell$ records the number of factors remaining to be expanded,
\[ x^pF(x)^\ell = fx^p F(x)^{\ell-1} + ax^{p+1}F(x)^{\ell-1}+bx^{p+2} F(x)^{\ell-1}. \]
We then define 
\begin{equation}
    Q = \mat{f & 0 & 0 \\ a & f & 0\\ b & a & f}, \quad C= \mat{0 & Q & 0 \\ 0 & 0 & Q \\ 0 & 0 & 0}, \quad D = I-C^T = \mat{I_3 & 0 & 0 \\ -Q^T & I_3 & 0 \\ 0 & -Q^T & I_3 }.
\end{equation}
This simple example demonstrates that the sparsity is not guaranteed after Carleman embedding. We see that $Q$ and $L_2$ are dense operators. 

In order to avoid dense operators, we  decompose the iteration of the Carleman system into several substeps and each step only contains sparse operators. We introduce sparse rectangular matrices $\Ein$ and $\Eout$ as 
\begin{equation}
    \Ein = I_3\otimes \bm e_0,\quad \Eout = \bm e_0\otimes I_3, \quad \bm e_0 = \mat{1 & 0 & 0}^T.
\end{equation}
For any fixed vector $\bm y\in\R^{3}$, we set the auxiliary vector 
\begin{equation}
    \bm w = \mat{\bm w^{(0)} \\ \bm w^{(1)} \\ \bm w^{(2)}},\quad \bm w^{(\ell)} = \mat{w_{0,\ell}\\ w_{1,\ell}\\ w_{2,\ell}}.
\end{equation}
to solve $D\bm w = \Eout \bm y$. Then $\bm w^{(0)} = \bm y$, $\bm w^{(1)} = Q^T \bm y$, $\bm w^{(2)} =Q^T \bm w^{(1)}. $ Therefore,
\begin{equation}
    \Ein^\dag \bm w = L_2\bm y.
\end{equation}
We conclude that 
\begin{equation}
    L_2 = \Ein^\dag (I-C^\dag)^{-1} \Eout.
\end{equation}

\subsection{General Construction and Recovery Identity}
We now show the same construction for the general vector iteration and prove the construction is equivalent to the original Carleman setup.

Recall that the Carleman blocks are defined in \cref{eq:carleman-blocks} with $L_{0,0}=1$ and $L_{0,k}=0$ for $k>0$. The truncated matrix $L_N$ acts on the enlarged Hilbert space 
\[ \mathcal{V}_N \coloneqq \bigoplus_{j=0}^N \H^{\otimes j}. \]
The dimension of $\mathcal{V}_N$ is $m = \sum_{j=0}^N d^j$. For each pair $(p,\ell)\in \{0,\cdots, N\}^2$, we introduce an auxiliary space $\H^{\otimes(p+\ell)}$, and we group the spaces as
\[ \mathcal{K}_\ell \coloneqq \bigoplus_{p=0}^N \H^{\otimes(p+\ell)}, \quad \mathcal{K}\coloneqq \bigoplus_{\ell=0}^N \mathcal{K}_\ell.\]
Then $\dim(\K_\ell) = d^\ell m$ and $\dim(\K)= m^2.$
We order the coordinates first by $\ell$ and then by $p$:
\[
\bm w\coloneqq\mat{\bm w^{(0)} \\ \vdots \\ \bm w^{(N)}},
\quad
\bm w^{(\ell)}=\mat{ \bm w_{0,\ell} \\ \vdots \\ \bm w_{N,\ell}},
\quad
\bm w_{p,\ell}\in\H^{\otimes(p+\ell)}.
\]
Here $p$ records the input degree already produced, and $\ell$
records the number of factors still to be expanded.
For $1\le\ell\le N$, define
$Q_\ell:\K_\ell\to\K_{\ell-1}$ through its blocks:
\begin{equation}\label{eq:sparse-Q}
(Q_\ell)_{p',p}
=
\begin{cases}
I_{\H}^{\otimes p}\otimes F_i^\dag
\otimes I_{\H}^{\otimes(\ell-1)},
& p'=p+i,\quad i\in\{0,1,2\},\\
0,&\text{otherwise},
\end{cases}
\qquad 0\le p,p'\le N.
\end{equation}
Thus branches with $p+i>N$ are omitted.
Each nonzero block maps
\[
\H^{\otimes(p+\ell)}
\longrightarrow
\H^{\otimes(p+i+\ell-1)},
\]
and
\[
Q_\ell\in
\mathbb C^{(d^{\ell-1}m)\times(d^\ell m)}.
\]
Under the natural tensor identifications, we have
\[
Q_\ell=Q_1\otimes I_{\H}^{\otimes(\ell-1)}.
\]

Define $C:\K\to\K$ and $D:\K\to\K$ by
\begin{equation}\label{eq:sparse-C}
C=
\mat{
0 & Q_1 & 0 & \cdots & 0\\
0 & 0 & Q_2 & \cdots &0\\
\vdots & \vdots & \ddots & \ddots & \vdots\\
0 & \cdots & 0 & 0 & Q_N\\
0 & \cdots & 0 & 0 & 0},
\qquad
D=I_{\K}-C^\dag.
\end{equation}
Notice that $C$ is $(N+1)$-nilpotent. 
Hence
\[
C^{N+1}=0,
\qquad
D^{-1}=\sum_{s=0}^N(C^\dag)^s.
\]
In particular, $D$ is invertible without any norm condition on $C$.

Finally, we introduce the embeddings $E_{\mathrm{out}},E_{\mathrm{in}}:\V_N\to\K$
\begin{equation}\label{eq:sparse-embeddings}
E_{\mathrm{out}}
=
\begin{pmatrix}
I_m\\
0_{(m^2-m)\times m}
\end{pmatrix},
\qquad
E_{\mathrm{in}}
=
\bigoplus_{\ell=0}^N
\begin{pmatrix}
I_{d^\ell}\\
0_{(m-1)d^\ell\times d^\ell}
\end{pmatrix}.
\end{equation}
Both matrices have dimensions $m^2\times m$.
The map $E_{\mathrm{out}}$ inserts $\bm y$ into layer zero,
while $E_{\mathrm{in}}$ inserts $\bm y_\ell$ into the $p=0$
block of layer $\ell$. Consequently,
\[
E_{\mathrm{in}}^\dag\bm w=(\bm w_{0,0},\ldots,\bm w_{0,N})^T.
\]

\begin{lemma} \label{lem: recovery_Carleman}
The truncated Carleman matrix $L_N$ can be written as
    \begin{equation}
        L_N = \Ein^\dag (I-C^\dag)^{-1} \Eout.
    \end{equation}
\end{lemma}
We provide the proof in \cref{pf: lem-recovery-Carleman}.

\begin{lemma}
\label{lem:sparsity-D}
Suppose every row and column of $A$ and $B$ has at most
$s_A$ and $s_B$ nonzero entries, respectively, and $\bm f$ has at most $s_f$ nonzero entries.
Then every row and column of $D=I_{\K}-C^\dag$ has at most
\begin{equation}\label{eq:D-sparsity}
    s_D\coloneqq 1+s_f+s_A+s_B
\end{equation}
nonzero entries, independently of the truncation degree $N$.
\end{lemma}
We provide the proof in \cref{pf: lem-sparsity-D}.

\section{Quantum Algorithms and Complexity}
\label{sec:rescaled-history}

We use the sparse representation to prepare a quantum state proportional
to $\bm x_n$ at a prescribed iteration $n\ge1$.
The construction has three steps: rescale the coefficients to control
conditioning, solve a padded history system, and select its physical block.

\subsection{Rescaling and uniform operator bounds}

Under \cref{ass: main}, choose a fixed radius $r$ such that
\begin{equation}\label{eq:hs-parameters}
    \rho<r<r_+,\qquad
    a=\frac{\rho}{r}<1,\qquad
    q=\frac cr+\alpha+\beta r<1.
\end{equation}
Since $r_-\le\rho<r_+$, such an $r$ exists and satisfies
\[
    q-1=\frac{h(r)-r}{r}
    =\frac{\beta}{r}(r-r_-)(r-r_+)<0.
\]
With $\bm z_t=\bm x_t/r$, the original iteration becomes
\begin{equation}\label{eq:hs-scaled-map}
    \bm z_{t+1}
    =\widetilde F_0+\widetilde F_1\bm z_t
      +\widetilde F_2\bm z_t^{\otimes2},
    \qquad
    \widetilde F_0=\bm f/r,\quad
    \widetilde F_1=A,\quad
    \widetilde F_2=rB.
\end{equation}
Let $\widetilde L_N,\widetilde C$, and
$\widetilde D=I_{\K}-\widetilde C^\dag$ be the matrices obtained by applying
the construction of Section~\ref{sec:sparse-realization} to these coefficients.
The embeddings $\Ein,\Eout$ and the operator $P=\Eout\Ein^\dag$ are unchanged.
In particular,
\begin{equation}\label{eq:hs-recovery}
    \widetilde L_N=\Ein^\dag\widetilde D^{-1}\Eout,
    \qquad
    \widetilde{\bm y}_{t+1}=\widetilde L_N\widetilde{\bm y}_t,
    \qquad
    \widetilde{\bm y}_{0,j}=(\bm x_0/r)^{\otimes j}.
\end{equation}
The rescaled parameters are
\[
    (\widetilde\alpha,\widetilde\beta,\widetilde c,\widetilde u)
    =(\alpha,r\beta,c/r,u/r).
\]
Direct substitution gives
\begin{align}
    \widetilde\Delta
    &=(1-\alpha)^2-4(r\beta)(c/r)=\Delta,\notag\\
    \widetilde r_\pm
    &=\frac{1-\alpha\pm\sqrt\Delta}{2r\beta}=\frac{r_\pm}{r},
    \qquad
    \widetilde\rho=\max\{u/r,r_-/r\}=\frac\rho r,
    \qquad \widetilde R=\frac{\widetilde\rho}{\widetilde r_+}=R.
    \label{eq:hs-assumption-invariance}
\end{align}
Together with $\widetilde u=u/r<r_+/r=\widetilde r_+$, these identities
verify \cref{ass: main} for the rescaled iteration.
Moreover, $\widetilde\rho=a<1$ and the sum of the coefficient bounds is $q<1$.

To relate the truncated trajectories, note that
$\widetilde F_i=r^{i-1}F_i$ for $i=0,1,2$.
Every tensor product contributing to $L_{j,k}$ has $j$ factors whose
input degrees sum to $k$, so its scale factor is $r^{k-j}$.
Consequently, with $V_r=\bigoplus_{j=0}^N r^{-j}I_{d^j}$,
\[
    \widetilde L_N=V_rL_NV_r^{-1},\qquad
    \widetilde{\bm y}_t
    =(V_rL_NV_r^{-1})^tV_r\bm y_0=V_r\bm y_t.
\]
In particular,
\begin{equation}\label{eq:hs-physical-recovery}
    r\widetilde{\bm y}_{t,1}=\bm y_{t,1},\qquad
    \frac{\widetilde{\bm y}_{t,1}}
         {\norm{\widetilde{\bm y}_{t,1}}}
    =\frac{\bm y_{t,1}}{\norm{\bm y_{t,1}}}
\end{equation}
whenever the physical block is nonzero. Thus the classical approximation
is recovered by multiplying by $r$, while its normalized quantum state
needs no inverse rescaling. The circuit accesses $\bm f/r,A,rB$ directly.

For later use, define
\begin{equation}\label{eq:hs-constants}
    \delta=\frac1{1-q},\qquad
    g=\frac1{\sqrt{1-q^2}},\qquad
    K_r=\delta g^2,\qquad
    M_r=\frac{g}{\sqrt{1-a^2}}.
\end{equation}
These quantities depend on the problem parameters and on $r$, but not on
$n$, $N$, or $d$ when the supplied norm bounds are held fixed.

\begin{lemma}
\label{lem:hs-uniform}
For every $N\ge1$, set $J=\widetilde D^{-1}\Eout$. Then
\begin{equation}\label{eq:hs-D-bounds}
    \norm{\widetilde C}\le q, \quad
    \norm{\widetilde D}\le1+q, \quad
    \norm{\widetilde D^{-1}}\le\delta, \quad
    \kappa(\widetilde D)\le(1+q)\delta, \quad
    \norm{J}\le g.
\end{equation}
Along the truncated trajectory in \eqref{eq:hs-recovery},
\begin{equation}\label{eq:hs-trajectory-bounds}
    \norm{\widetilde{\bm y}_{t,j}}\le a^j,
    \qquad
    \norm{\widetilde{\bm y}_t}\le\frac1{\sqrt{1-a^2}},
    \qquad
    \norm{J\widetilde{\bm y}_t}\le M_r
\end{equation}
for all $t\ge0$ and $0\le j\le N$.
Moreover, the truncated update is uniformly power bounded:
\begin{equation} \label{eq: L_N_power_bound}
    \norm{\widetilde L_N^k}\le g
    \qquad\forall k\in\N_0.
\end{equation}
\end{lemma}

We defer the proof to \cref{pf: lem-hs-uniform}.

\subsection{The sparse history system}

Having constructed a sparse realization of the discrete Carleman recurrence, we now employ the history-state framework of \cite{krovi2023improved}. We encode the auxiliary trajectory in a linear system, analyze its conditioning, and subsequently use terminal padding to improve the probability of extracting the final state. For our discrete-time construction, we establish conditioning bounds whose constants are independent of the truncation degree.

Define the auxiliary vectors by
\begin{equation}\label{eq:hs-aux-trajectory}
    \widetilde{\bm w}_t=J\widetilde{\bm y}_t,
    \qquad
    \widetilde{\bm y}_{t+1}=\Ein^\dag\widetilde{\bm w}_t.
\end{equation}
Multiplying by $\widetilde D$ gives
\[
    \widetilde D\widetilde{\bm w}_0=\Eout\widetilde{\bm y}_0,
    \qquad
    \widetilde D\widetilde{\bm w}_t
      =\Eout\widetilde{\bm y}_t
      =\Eout\Ein^\dag\widetilde{\bm w}_{t-1}
      =P\widetilde{\bm w}_{t-1}\quad(t\ge1).
\]
Let $S_n|t\rangle=|t+1\rangle$ for $0\le t<n-1$ and
$S_n|n-1\rangle=0$. Stacking the first $n$ equations gives
\begin{equation}\label{eq:hs-history}
    \widetilde{\mathcal A}_n
    \begin{pmatrix}\widetilde{\bm w}_0\\\vdots\\
    \widetilde{\bm w}_{n-1}\end{pmatrix}
    =|0\rangle\otimes\Eout\widetilde{\bm y}_0,
    \qquad
    \widetilde{\mathcal A}_n
    =I_n\otimes\widetilde D-S_n\otimes P.
\end{equation}
Here $\dim\K=m^2$, with $m=\sum_{j=0}^Nd^j$. The complexity of solving \cref{eq:hs-history} by using QLSA would require the condition number of $\widetilde{\mathcal{A}}_n. $

\begin{lemma}[History conditioning, uniformly in the truncation degree]
\label{lem:hs-history}
For all $n,N\ge1$,
\begin{align}
    \norm{\widetilde{\mathcal A}_n}
    &\le1+\sqrt{1+q^2},\label{eq:hs-history-norm}\\
    \norm{\widetilde{\mathcal A}_n^{-1}}
    &\le\delta\bigl(1+(n-1)g^2\bigr)\le nK_r,
    \label{eq:hs-history-inverse}\\
    \frac n{\sqrt3}
    \le\kappa(\widetilde{\mathcal A}_n)
    &\le\bigl(1+\sqrt{1+q^2}\bigr)nK_r.
    \label{eq:hs-history-condition}
\end{align}
In particular, the constants in
$\kappa(\widetilde{\mathcal A}_n)=\Theta(n)$ are independent of $N$.
\end{lemma}
We defer the proof to \cref{pf: lem-hs-history}.

\begin{remark}
    An alternative approach is to reformulate the discrete iteration system as a suitable continuous-time evolution and apply Schr\"odingerization \cite{jin2024discrete}. Here, we choose to use QLSA, which applies directly to the sparse history system and exploits the conditioning bounds established above.
\end{remark}

\subsection{Terminal padding and physical-state extraction}

\paragraph{Extraction from the original history.}
Conditioned on preparing the ideal normalized solution of
\eqref{eq:hs-history}, the history state is
\begin{equation}\label{eq:hs-unpadded-state}
    |W\rangle
    =\frac{1}{\sqrt{H_n}}
      \sum_{t=0}^{n-1}|t\rangle\otimes\widetilde{\bm w}_t,
    \qquad
    H_n=\sum_{t=0}^{n-1}\norm{\widetilde{\bm w}_t}^2.
\end{equation}
Since $\Ein^\dag\widetilde{\bm w}_{n-1}=\widetilde{\bm y}_n$,
the desired physical vector $\widetilde{\bm y}_{n,1}$ occupies
the auxiliary sector $(p,\ell)=(0,1)$ at clock value $n-1$.
Writing $h_n=\norm{\widetilde{\bm y}_{n,1}}$, the probability
of selecting this clock value and this sector is
\begin{equation}\label{eq:hs-unpadded-success}
    p_{\mathrm{raw}}
    =\frac{\norm{\widetilde{\bm w}_{n-1}}^2}{H_n}
      \frac{h_n^2}{\norm{\widetilde{\bm w}_{n-1}}^2}
    =\frac{h_n^2}{H_n}.
\end{equation}
The first factor is the probability of selecting the terminal clock, and the second is the conditional probability of selecting its physical
sector.

The degree-zero Carleman coordinate remains equal to one, and $J$
preserves it in auxiliary sector $(0,0)$. Hence
$\norm{\widetilde{\bm w}_t}\ge1$.
Together with $\norm{\widetilde{\bm w}_t}\le M_r$ from
\cref{lem:hs-uniform}, this gives
\begin{equation}\label{eq:hs-unpadded-probability-bounds}
    n\le H_n\le nM_r^2,
    \qquad
    \frac{h_n^2}{nM_r^2}
    \le p_{\mathrm{raw}}
    \le\frac{h_n^2}{n}.
\end{equation}
Thus extraction from a single terminal clock carries an explicit
$1/n$ factor, in addition to the dependence on $h_n$.

\paragraph{Terminal padding and physical-state extraction}

To remove the factor $1/n$ in \eqref{eq:hs-unpadded-probability-bounds}, we use the terminal-state padding construction
of \cite{berry2017quantum,liu2021efficient}.
Introduce $n$ additional blocks satisfying
\[
    \bm v_1=\widetilde{\bm w}_{n-1},
    \qquad
    \bm v_j=\bm v_{j-1},
    \quad 2\le j\le n.
\]
These equations repeat the terminal auxiliary vector without advancing
the original dynamics.
The resulting unnormalized padded history vector is
\begin{equation}\label{eq:hs-padded-vector}
    \bm\Omega
    =\sum_{t=0}^{n-1}|t\rangle\otimes\widetilde{\bm w}_t
     +\sum_{t=n}^{2n-1}|t\rangle
        \otimes\widetilde{\bm w}_{n-1}.
\end{equation}
In particular, $\bm\Omega$ contains $2n$ blocks: the original $n$
auxiliary vectors followed by $n$ additional copies of the terminal
vector.

The padded linear system is
\begin{equation}\label{eq:hs-padded-system}
    \mathcal B_n\bm\Omega=\bm b,
    \qquad
    \bm b=|0\rangle\otimes\Eout\widetilde{\bm y}_0,
\end{equation}
where $\mathcal B_n$ is block lower-bidiagonal, with nonzero blocks
\begin{equation}\label{eq:hs-padded-matrix}
    (\mathcal B_n)_{t,t}
    =\begin{cases}
        \widetilde D,&0\le t<n,\\
        I_{\K},&n\le t<2n,
      \end{cases}
    \qquad
    (\mathcal B_n)_{t,t-1}
    =\begin{cases}
        -P,&1\le t<n,\\
        -I_{\K},&n\le t<2n.
      \end{cases}
\end{equation}

We now extract the physical state from the normalized padded history
$|\Omega\rangle=\bm\Omega/\norm{\bm\Omega}$.
Let $\Pi$ be the orthogonal projector selecting the appended labels
$n,\ldots,2n-1$ and auxiliary sector $(p,\ell)=(0,1)$.
Every selected block contains $\widetilde{\bm y}_{n,1}$, so
\[
    \Pi\bm\Omega
    =\sum_{t=n}^{2n-1}|t\rangle
      \otimes|p=0,\ell=1\rangle
      \otimes\widetilde{\bm y}_{n,1},
    \qquad
    \norm{\Pi\bm\Omega}^2=nh_n^2.
\]
Meanwhile, orthogonality of the history labels gives
\[
    \norm{\bm\Omega}^2
    =H_n+n\norm{\widetilde{\bm w}_{n-1}}^2
    \le 2nM_r^2.
\]
Consequently, conditioned on preparing the ideal normalized padded
history, the extraction probability is
\begin{equation}\label{eq:hs-success}
    p_{\mathrm{out}}
    =\frac{\norm{\Pi\bm\Omega}^2}{\norm{\bm\Omega}^2}
    =\frac{nh_n^2}
           {H_n+n\norm{\widetilde{\bm w}_{n-1}}^2}
    \ge\frac{h_n^2}{2M_r^2}.
\end{equation}
Padding therefore removes the explicit $1/n$ factor from the
probability bound, while retaining the dependence on $h_n$.

For $h_n>0$, the normalized state after successful projection is
\begin{equation}\label{eq:hs-projected-state}
    \frac{\Pi\bm\Omega}{\norm{\Pi\bm\Omega}}
    =\left(\frac{1}{\sqrt n}
       \sum_{t=n}^{2n-1}|t\rangle\right)
      \otimes|p=0,\ell=1\rangle
      \otimes\frac{\widetilde{\bm y}_{n,1}}{h_n}.
\end{equation}
Discarding the history and auxiliary-label registers leaves
$\widetilde{\bm y}_{n,1}/h_n$, which equals the normalized unscaled
truncated physical vector by \eqref{eq:hs-physical-recovery}.

\subsection{Quantum implementation and complexity}

\begin{assumption}[Quantum input access]\label{ass:hs-access}
We have coherent row, column, and value oracles for $\bm f,A,B$,
with controlled and inverse access. The parameters in
\cref{ass: main} are known. For $u=\norm{\bm x_0}>0$,
a unitary $U_x$ prepares $\bm x_0/u$.
\end{assumption}

\begin{lemma}[Preparation of the padded history]
\label{lem:hs-implementation}
Under \cref{ass:hs-access}, set
\begin{equation}\label{eq:hs-sparsity}
    s=2+s_f+s_A+s_B.
\end{equation}
The padded matrix satisfies
\begin{equation}\label{eq:hs-padded-bounds}
    \norm{\mathcal B_n}\le2+q,
    \qquad
    \norm{\mathcal B_n^{-1}}\le2nK_r.
\end{equation}
For $0<\eta\le1/2$, a coherent QLSA prepares an $\eta$-approximation in vector norm to
\[
    |\Omega\rangle=\frac{\bm\Omega}{\norm{\bm\Omega}}
\]
with constant heralded success probability, using
\[
    O\!\left(snK_r\log(1/\eta)\right)
\]
queries to the coefficient oracles and a unitary $U_b$
preparing $\bm b/\norm{\bm b}$.
Each use of $U_b$ requires at most $N$ controlled calls to $U_x$,
and none when $u=0$.
\end{lemma}
The proof is provided in \cref{pf: lem-hs-implementation}.

\begin{theorem}[Quantum preparation of the physical state]
\label{thm:hs-quantum}
Suppose \cref{ass: main,ass:hs-access} hold.
Fix $n\ge1$, $0<\varepsilon\le1$, and a known lower bound
$0<\xi\le\norm{\bm x_n}$.
Choose $r$ as in \eqref{eq:hs-parameters}, and set
\begin{equation}\label{eq:hs-N-choice}
    N=\max\left\{1,
       \left\lceil
       \frac{\log(8\rho/(\varepsilon\xi))}
            {\log(1/R)}
       \right\rceil\right\},
    \qquad
    \Gamma=\frac{2rM_r}{\xi}.
\end{equation}
There is an algorithm with success probability at least
$2/3$ whose conditional output $\sigma_{\mathrm{out}}$ satisfies
\begin{equation}\label{eq:hs-output-error}
    \frac12\left\|
       \sigma_{\mathrm{out}}
       -\frac{\bm x_n\bm x_n^\dag}{\norm{\bm x_n}^2}
    \right\|_1\le\varepsilon.
\end{equation}
Its query complexity is
\begin{equation}\label{eq:hs-query-bound}
    Q=O\!\left(
       snK_r\Gamma\log\frac{32\Gamma}{\varepsilon}
       \right)
     =\widetilde O\!\left(
       \frac{sn\Gamma}{(1-q)(1-q^2)}
       \right) = \widetilde{O}\!\left(\frac{snr}{\xi(1-q)(1-q^2)^{3/2}(1-a^2)^{1/2}}\right).
\end{equation}
to the coefficient oracles and $U_b$, including controlled
and inverse queries. The number of $U_x$ queries is $O(NQ)$,
and is zero when $u=0$.
\end{theorem}
The proof is provided in \cref{pf: thm-hs-quantum}.

The lower bound $\xi$ is an additional input promise.
For fixed stability parameters and $\xi$, the query bound is linear in $n$. The output encodes the direction of $\bm x_n$. Classical readout requires additional estimation. 

\begin{remark} \label{rem: sparsity_speedup}
    To the best of our knowledge, recent works \cite{bravo2026optimal, berry2026root} improve the sparsity dependence of the QLSA query complexity from linear in $s$ to proportional to $\sqrt{s}$.
\end{remark}

\section{Numerical Experiments} \label{sec: num}
This section presents numerical evidence for the convergence of the truncated Carleman embedding.

We first consider the one-dimensional quadratic map
\begin{equation} \label{eq:example-1}
    x_{n+1} = 0.1 + 0.4x_n + 0.5x_n^2,
    \qquad x_0 > 0.
\end{equation}
The map in \cref{eq:example-1} has two fixed points, $x=0.2$ and $x=1$. The fixed point $x=0.2$ is stable, whereas $x=1$ is unstable. In particular, for the initial conditions considered below, the iterates converge to $x=0.2$ rather than $x=1$.

In \cref{fig:carleman-results}, we consider
\[
x_0 \in \{0.25,\; 0.5,\; 0.75,\; 0.9\}.
\]
\Cref{fig:error-vs-N} shows that the Carleman truncation error decreases as the truncation level $N$ increases. The y-axis plots the maximum of absolute error over all iteration $n\le 100$. Moreover, \cref{fig:physical-trajectories} shows that, at truncation level $N=12$, the Carleman trajectories corresponding to these initial conditions approach the stable fixed point $x=0.2$.

\begin{figure}[htbp]
    \centering
    \begin{subfigure}[b]{0.49\textwidth}
        \centering
        \includegraphics[width=\linewidth]{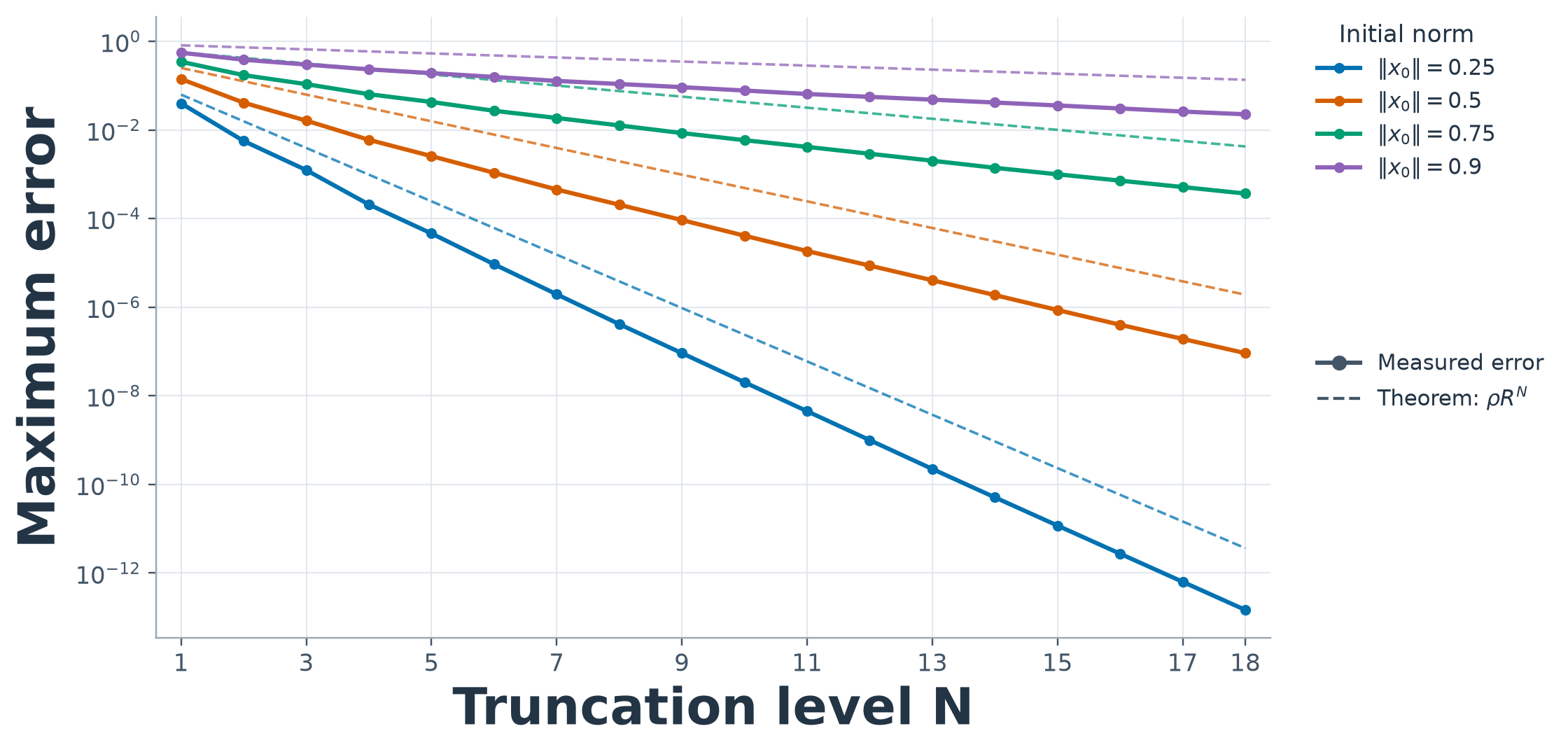}
        \caption{Carleman truncation error versus truncation level $N$.}
        \label{fig:error-vs-N}
    \end{subfigure}
    \hfill
    \begin{subfigure}[b]{0.49\textwidth}
        \centering
        \includegraphics[width=\linewidth]{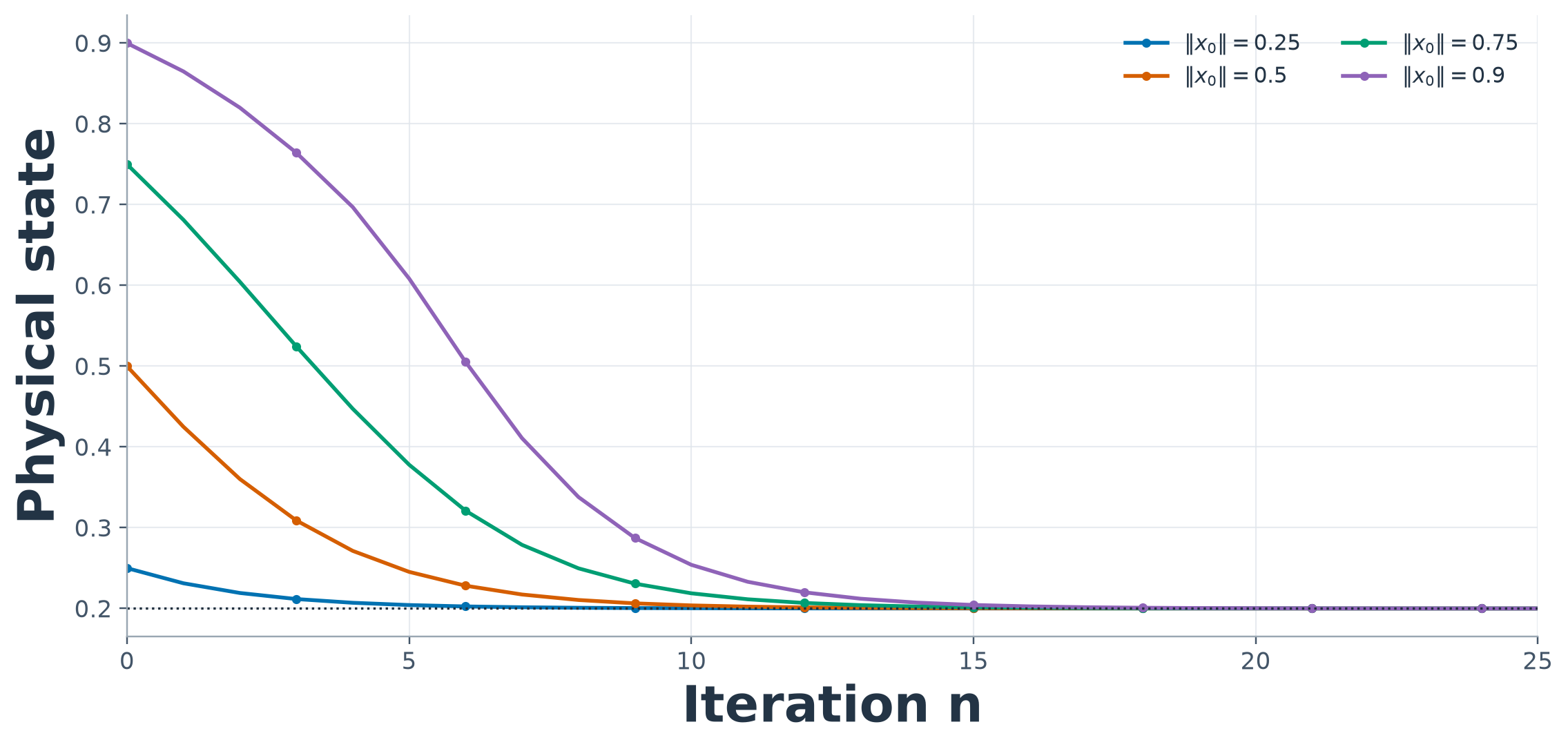}
        \caption{Trajectories from the finite Carleman system toward the fixed point with truncation level $N=12$.}
        \label{fig:physical-trajectories}
    \end{subfigure}
    \caption{Numerical results for the discrete Carleman approximation.}
    \label{fig:carleman-results}
\end{figure}

We next consider a three-dimensional example from \cite{gleich2015}. Its discrete iteration is
\begin{equation}\label{eq:example-2}
    \bm x_{n+1}=(1-\eta)\bm v+\eta T(\bm x_n\otimes\bm x_n),
\end{equation}
where
\[ T= \mat{0 & \frac12 & 0 & \frac12 & 0 & 1 & \frac12 & \frac12 & 0\\ 0 & 0 & 0 & 0 & \frac12 & 0 & 0 & \frac12 & 0\\ 1 & \frac12 & 1 & \frac12 & \frac12 & 0 & \frac12 & 0 & 1 }. \]
We set $\eta=0.2$ and $\bm v=\frac13\mat{1&1&1}^{T}$, and initialize the iteration at $\bm x_0=\bm v$. These parameters satisfy \cref{ass: main}. \Cref{fig:example2} shows the Carleman truncation error for several truncation levels $N$.

Writing \cref{eq:example-2} as $\bm x_{n+1}=\bm f+B\bm x_n^{\otimes 2}$, we have $\bm f=(1-\eta)\bm v$, $A=0$, and $B=\eta T$. Because this map contains only constant and quadratic terms, odd input-degree blocks do not contribute to the Carleman update. Consequently, truncation at levels $2k$ and $2k+1$ gives the same physical-state approximation for every $k\geq1$. This explains why the error in \cref{fig:example2} decreases from odd to even $N$, but remains unchanged from even to odd $N$.

\begin{figure}[htbp]
    \centering
    \includegraphics[width=0.5\linewidth]{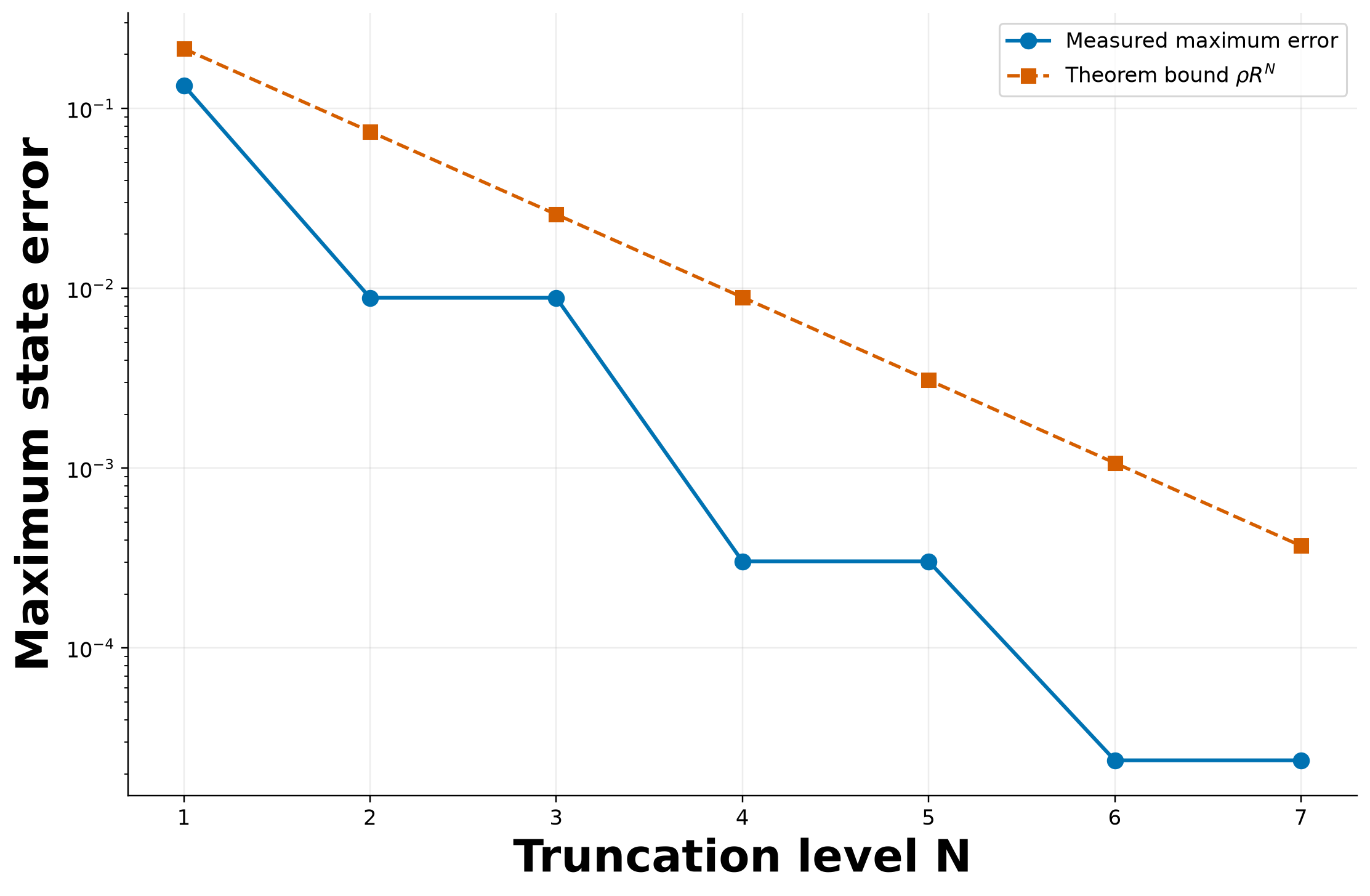}
    \caption{The Carleman truncation error versus truncation level $N$.}
    \label{fig:example2}
\end{figure}

\section{Conclusion}

We studied the Carleman embedding of forced quadratic discrete dynamical systems. Under explicit stability and initial-state assumptions, the degree-$N$ truncation approximates the original trajectory with an error that decreases geometrically in $N$, uniformly over all iterations. In particular, for fixed stability parameters, a degree logarithmic in the inverse accuracy suffices, regardless of how many steps the system is evolved.

A finite Carleman lift is not automatically efficient to implement: even
sparse coefficient maps can produce dense blocks when their tensor
products are formed. We addressed this by representing the truncated
Carleman update exactly through an auxiliary operator whose layers
apply the coefficient maps one tensor factor at a time. The resulting
linear system has row and column sparsities bounded independently of
the truncation degree. This separates the accuracy needed for the
Carleman approximation from the sparsity of the system used to
compute it.

We then rescaled the auxiliary representation and assembled the
updates into a history system. Its condition number grows as
$\Theta(n)$, with constants independent of the truncation degree.
Combined with the uniform error bound, this gives a quantum algorithm
for preparing a state close to the normalized $n$th iterate. Under the
stated coherent-access assumptions and a positive lower bound on the
norm of that iterate, the query complexity is linear in $n$ and
polylogarithmic in the inverse accuracy when the stability parameters
are fixed. The result concerns state preparation. Estimating
observables or recovering the full iterate classically requires
additional queries. Finally, the numerical experiments illustrate the predicted decrease in truncation error for one- and three-dimensional examples. 

In our work, we considered the polynomial discrete systems with focus on quadratic case. On the other hand, there are plenty of nonlinear discrete systems are not in polynomial in machine learning, optimization and physical models. Whether Carleman embedding guarantees error convergence for  broader classes of functions is an interesting future direction. Extending the stability and error analysis, estimating the preparation and measurement cost will also help establish end to end cost.

\section{Acknowledgments}
The authors gratefully acknowledge support from the National Science Foundation under Grants Nos.~DMS-2411120 and DMS-2552687. The authors fully acknowledge the use of LLM for polishing the language and implementing the numerical results.


\bibliographystyle{unsrtnat}
\bibliography{qde-references}

@article{sanavio2024lattice,
  author  = {Sanavio, Claudio and Succi, Sauro},
  title   = {{Lattice Boltzmann--Carleman} quantum algorithm and circuit
             for fluid flows at moderate {Reynolds} number},
  journal = {AVS Quantum Science},
  volume  = {6},
  pages   = {023802},
  year    = {2024},
  doi     = {10.1116/5.0195549}
}

@article{sanavio2025carleman,
  author  = {Sanavio, Claudio and Simon, William A. and Ralli, Alex
             and Love, Peter and Succi, Sauro},
  title   = {{Carleman-lattice-Boltzmann} quantum circuit
             with matrix access oracles},
  journal = {Physics of Fluids},
  volume  = {37},
  pages   = {037123},
  year    = {2025},
  doi     = {10.1063/5.0254588}
}

@article{jin2024discrete,
    author = {Jin, Shi and Liu, Nana},
    title = {Quantum simulation of discrete linear dynamical systems and simple iterative methods in linear algebra},
    journal = {Proceedings of the Royal Society A: Mathematical, Physical and Engineering Sciences},
    volume = {480},
    number = {2291},
    pages = {20230370},
    year = {2024},
    month = {06},
    issn = {1364-5021},
    doi = {10.1098/rspa.2023.0370}
}

@article{pruekprasert2020moment,
  author = {Pruekprasert, Sasinee and Takisaka, Toru and Eberhart, Clovis and
            Cetinkaya, Ahmet and Dubut, J{\'e}r{\'e}my},
  title = {Moment Propagation of Discrete-Time Stochastic Polynomial Systems using Truncated {Carleman} Linearization},
  journal = {IFAC-PapersOnLine},
  volume = {53},
  number = {2},
  pages = {14462--14469},
  year = {2020},
  doi = {10.1016/j.ifacol.2020.12.1447},
  url = {https://doi.org/10.1016/j.ifacol.2020.12.1447}
}

@article{pruekprasert2024moment,
  author = {Pruekprasert, Sasinee and Dubut, Jeremy and Takisaka, Toru and
            Eberhart, Clovis and Cetinkaya, Ahmet},
  title = {Moment propagation of polynomial systems through {Carleman} linearization for probabilistic safety analysis},
  journal = {Automatica},
  volume = {160},
  pages = {111441},
  year = {2024},
  doi = {10.1016/j.automatica.2023.111441},
  eprint = {2201.08648},
  archivePrefix = {arXiv},
  primaryClass = {eess.SY},
  url = {https://doi.org/10.1016/j.automatica.2023.111441}
}

@article{myszkowski2022recursions,
  author = {Myszkowski, Miko{\l}aj},
  title = {{Carleman} Linearization and Systems of Arbitrary Depth Polynomial Recursions},
  journal = {Advances in Linear Algebra \& Matrix Theory},
  volume = {12},
  number = {1},
  pages = {1--23},
  year = {2022},
  doi = {10.4236/alamt.2022.121001},
  eprint = {2109.13165},
  archivePrefix = {arXiv},
  url = {https://doi.org/10.4236/alamt.2022.121001}
}

@inproceedings{bai2019deep,
  title = {Deep Equilibrium Models},
  author = {Bai, Shaojie and Kolter, J. Zico and Koltun, Vladlen},
  booktitle = {Advances in Neural Information Processing Systems},
  volume = {32},
  year = {2019},
  url = {https://proceedings.neurips.cc/paper/2019/hash/01386bd6d8e091c2ab4c7c7de644d37b-Abstract.html}
}

@inproceedings{bai2021stabilizing,
  title = {Stabilizing Equilibrium Models by {Jacobian} Regularization},
  author = {Bai, Shaojie and Koltun, Vladlen and Kolter, J. Zico},
  booktitle = {Proceedings of the 38th International Conference on Machine Learning},
  series = {Proceedings of Machine Learning Research},
  volume = {139},
  pages = {554--565},
  year = {2021},
  publisher = {PMLR},
  url = {https://proceedings.mlr.press/v139/bai21b.html}
}

@article{berry2026root,
  title = {Root-sparsity scaling for Hamiltonian simulation, differential equations and linear-system solvers},
  author = {Berry, Dominic W and Rawat, Abhishek and Simon, Sophia and Low, Guang Hao},
  journal = {arXiv preprint arXiv:2610.02939},
  year = {2026},
  url = {https://arxiv.org/abs/2610.02939}
}

@article{bravo2026optimal,
  title = {{An Optimal Quantum Linear Systems Algorithm}},
  author = {Bravo-Prieto, Carlos and Harrow, Aram W. and Kothari, Robin},
  journal = {arXiv preprint arXiv:2609.35660},
  year = {2026},
  url = {https://arxiv.org/abs/2609.35660}
}

@article{gleich2015,
  author = {Gleich, David F. and Lim, Lek-Heng and Yu, Yongyang},
  title = {Multilinear PageRank},
  journal = {SIAM Journal on Matrix Analysis and Applications},
  volume = {36},
  number = {4},
  pages = {1507-1541},
  year = {2015},
  doi = {10.1137/140985160}
}

@inproceedings{lee2016gradient,
  author = {Lee, Jason D. and Simchowitz, Max and
               Jordan, Michael I. and Recht, Benjamin},
  title = {Gradient Descent Only Converges to Minimizers},
  booktitle = {29th Annual Conference on Learning Theory},
  series = {Proceedings of Machine Learning Research},
  volume = {49},
  pages = {1246--1257},
  publisher = {PMLR},
  year = {2016},
  url = {https://proceedings.mlr.press/v49/lee16.html}
}

@inproceedings{pascanu2013difficulty,
  author = {Pascanu, Razvan and Mikolov, Tomas and Bengio, Yoshua},
  title = {On the difficulty of training recurrent neural networks},
  booktitle = {Proceedings of the 30th International Conference
               on Machine Learning},
  series = {Proceedings of Machine Learning Research},
  volume = {28},
  number = {3},
  pages = {1310--1318},
  publisher = {PMLR},
  year = {2013},
  url = {https://proceedings.mlr.press/v28/pascanu13.html}
}

@article{carleman1932application,
  author = {Carleman, Torsten},
  title = {Application de la th{\'e}orie des {\'e}quations
             int{\'e}grales lin{\'e}aires aux syst{\`e}mes
             d'{\'e}quations diff{\'e}rentielles non lin{\'e}aires},
  journal = {Acta Mathematica},
  volume = {59},
  pages = {63--87},
  year = {1932},
  doi = {10.1007/BF02546499}
}

@article{harrow2009quantum,
  author = {Harrow, Aram W. and Hassidim, Avinatan and Lloyd, Seth},
  title = {Quantum Algorithm for Linear Systems of Equations},
  journal = {Physical Review Letters},
  volume = {103},
  pages = {150502},
  year = {2009},
  doi = {10.1103/PhysRevLett.103.150502}
}

@article{berry2017quantum,
  author = {Berry, Dominic W. and Childs, Andrew M. and Ostrander, Aaron and Wang, Guoming},
  title = {Quantum Algorithm for Linear Differential Equations with Exponentially Improved Dependence on Precision},
  journal = {Communications in Mathematical Physics},
  volume = {356},
  pages = {1057--1081},
  year = {2017},
  doi = {10.1007/s00220-017-3002-y}
}

@misc{lloyd2020nonlinear,
  author = {Lloyd, Seth and De Palma, Giacomo and Gokler, Can and Kiani, Bobak and Liu, Zi-Wen and Marvian, Milad and Tennie, Felix and Palmer, Tim},
  title = {Quantum algorithm for nonlinear differential equations},
  year = {2020},
  eprint = {2011.06571},
  archiveprefix = {arXiv},
  primaryclass = {quant-ph},
  note = {arXiv:2011.06571},
  url = {https://arxiv.org/abs/2011.06571}
}

@article{joseph2020koopman,
  author = {Joseph, Ilon},
  title = {{Koopman--von Neumann} approach to quantum simulation of nonlinear classical dynamics},
  journal = {Physical Review Research},
  volume = {2},
  pages = {043102},
  year = {2020},
  doi = {10.1103/PhysRevResearch.2.043102}
}

@article{berkolaiko1998analysis,
  author = {Berkolaiko, G. and Rabinovich, S. and Havlin, S.},
  title = {Analysis of {Carleman} Representation of Analytical Recursions},
  journal = {Journal of Mathematical Analysis and Applications},
  volume = {224},
  number = {1},
  pages = {81--90},
  year = {1998},
  doi = {10.1006/jmaa.1998.5986}
}

@article{liu2021efficient,
  author = {Liu, Jin-Peng and Kolden, Herman {\O}ie and Krovi, Hari K. and Loureiro, Nuno F. and Trivisa, Konstantina and Childs, Andrew M.},
  title = {Efficient quantum algorithm for dissipative nonlinear differential equations},
  journal = {Proceedings of the National Academy of Sciences},
  volume = {118},
  pages = {e2026805118},
  year = {2021},
  doi = {10.1073/pnas.2026805118}
}

@article{krovi2023improved,
  author = {Krovi, Hari},
  title = {Improved quantum algorithms for linear and nonlinear differential equations},
  journal = {Quantum},
  volume = {7},
  pages = {913},
  year = {2023},
  doi = {10.22331/q-2023-02-02-913}
}

@article{liu2024towards,
  author = {Liu, Junyu and Liu, Minzhao and Liu, Jin-Peng and Ye, Ziyu and Wang, Yunfei and Alexeev, Yuri and Eisert, Jens and Jiang, Liang},
  title = {Towards provably efficient quantum algorithms for large-scale machine-learning models},
  journal = {Nature Communications},
  volume = {15},
  pages = {434},
  year = {2024},
  doi = {10.1038/s41467-023-43957-x}
}

@article{jennings2026end,
  author = {Jennings, David and Korzekwa, Kamil and Lostaglio, Matteo and Ashworth, Richard and Marsili, Emanuele and Rolston, Stephen},
  title = {End-to-End Quantum Algorithm for Nonlinear Fluid Dynamics with Bounded Quantum Advantage},
  journal = {PRX Quantum},
  volume = {7},
  number = {3},
  pages = {033060},
  year = {2026},
  doi = {10.1103/xysy-q3fp}
}

@inproceedings{gilyen2019quantum,
  title = {Quantum Singular Value Transformation and Beyond: Exponential Improvements for Quantum Matrix Arithmetics},
  author = {Gily{\'e}n, Andr{\'a}s and Su, Yuan and Low, Guang Hao and Wiebe, Nathan},
  booktitle = {Proceedings of the 51st Annual ACM SIGACT Symposium on Theory of Computing},
  pages = {193--204},
  year = {2019},
  publisher = {Association for Computing Machinery},
  doi = {10.1145/3313276.3316366},
  eprint = {1806.01838},
  archiveprefix = {arXiv}
}

@article{yoder2014fixedpoint,
  title = {Fixed-Point Quantum Search with an Optimal Number of Queries},
  author = {Yoder, Theodore J. and Low, Guang Hao and Chuang, Isaac L.},
  journal = {Physical Review Letters},
  volume = {113},
  number = {21},
  pages = {210501},
  year = {2014},
  doi = {10.1103/PhysRevLett.113.210501},
  eprint = {1409.3305},
  archiveprefix = {arXiv},
  primaryclass = {quant-ph}
}

@article{Ber14,
  title = {High-order quantum algorithm for solving linear differential equations},
  author = {Berry, Dominic W.},
  journal = {Journal of Physics A: Mathematical and Theoretical},
  volume = {47},
  number = {10},
  pages = {105301},
  year = {2014},
  publisher = {IOP Publishing},
  doi = {10.1088/1751-8113/47/10/105301}
}

@article{childs2020quantum,
  title = {Quantum spectral methods for differential equations},
  author = {Childs, Andrew M and Liu, Jin-Peng},
  journal = {Communications in Mathematical Physics},
  volume = {375},
  number = {2},
  pages = {1427--1457},
  year = {2020},
  publisher = {Springer},
  doi = {10.1007/s00220-020-03699-z}
}

@article{costa2022optimal,
  author = {Costa, Pedro C. S. and An, Dong and Sanders, Yuval R. and Su, Yuan and Babbush, Ryan and Berry, Dominic W.},
  title = {Optimal Scaling Quantum Linear-Systems Solver via Discrete Adiabatic Theorem},
  journal = {PRX Quantum},
  volume = {3},
  pages = {040303},
  year = {2022},
  doi = {10.1103/PRXQuantum.3.040303}
}

@article{may1976simple,
  author = {May, Robert M.},
  title = {Simple mathematical models with very complicated dynamics},
  journal = {Nature},
  volume = {261},
  pages = {459--467},
  year = {1976},
  doi = {10.1038/261459a0}
}

@article{xue2021homotopy,
  author = {Xue, Cheng and Wu, Yu-Chun and Guo, Guo-Ping},
  title = {Quantum homotopy perturbation method for nonlinear dissipative ordinary differential equations},
  journal = {New Journal of Physics},
  volume = {23},
  pages = {123035},
  year = {2021},
  doi = {10.1088/1367-2630/ac3eff}
}

@article{xue2022quadratic,
  author = {Xue, Cheng and Xu, Xiao-Fan and Wu, Yu-Chun and Guo, Guo-Ping},
  title = {Quantum algorithm for solving a quadratic nonlinear system of equations},
  journal = {Physical Review A},
  volume = {106},
  pages = {032427},
  year = {2022},
  doi = {10.1103/PhysRevA.106.032427}
}

@article{childs2017quantum,
  title = {Quantum Algorithm for Systems of Linear Equations with Exponentially Improved Dependence on Precision},
  author = {Childs, Andrew M. and Kothari, Robin and Somma, Rolando D.},
  journal = {SIAM Journal on Computing},
  volume = {46},
  number = {6},
  pages = {1920--1950},
  year = {2017},
  doi = {10.1137/16M1087072}
}


\crefalias{section}{appendix}
\crefname{appendix}{Appendix}{Appendices}
\Crefname{appendix}{Appendix}{Appendices}

\appendix

\section{Proof of \cref{thm: uniform}} \label{pf: thm-uniform}

\begin{proof}
Define 
\begin{equation}
    h(s)\coloneqq c+\alpha s + \beta s^2.
\end{equation}
Under \cref{ass: main}, we know $\Delta>0$. Thus, to find the fixed point of $h$, we have
\[ h(s)=s\To s= \frac{(1-\alpha)\pm\sqrt{\Delta}}{2\beta} = r_{\pm}. \]
Here $r_-<r_+$, $\rho\coloneqq\max\{\norm{\bm x_0}, r_-\}$. $R\coloneqq\frac{\rho}{r_+}<1$. By triangular inequality, we have
\[ \norm{F(\bm x)} \le \norm{\bm f}+\norm{A}\norm{\bm x}+\norm{B}\norm{\bm x}^2 \le h(\norm{\bm x}). \]
Note that $h'(s)\ge0$ for nonnegative $s$, thus $h(s)$ is increasing for $s\ge0$. Therefore, if $\norm{\bm x_n}\le \rho$, 
\[ \norm{\bm x_{n+1}} =\norm{F(\bm x_n)} \le h(\norm{\bm x_n})\le h(\rho)\le \rho. \]
The last inequality follows from $r_-\le \rho <r_+$. Since $\norm{\bm x_0}\le\rho$, we conclude that $\norm{\bm x_n}\le\rho$ for all $n\in \N_0$. 

For any $j\in\N$, let 
\[ h(s)^j = \sum_{k=0}^{2j}m_{j,k}s^k, \]
for some coefficients $m_{j,k}$. Since $\alpha$ and $\beta$ are upper norm bounds of $A$ and $B$, $\norm{L_{j,k}}\le m_{j,k}$ for all $j,k$. Thus,
\[\sum_{k=0}^{2j}\norm{L_{j,k}\bm x_n^{\otimes k}}\le \sum_{k=0}^{2j}\norm{L_{j,k}}\norm{\bm x_n^{\otimes k}}\le \sum_{k=0}^{2j} m_{j,k}r_+^k = 
h(r_+)^j = r_+^j. \]
Define 
\[
    \bm \eta_{n,j} = \bm x_n^{\otimes j} - \bm y_{n,j}, \qquad \bm \eta_{n,0}=0, \qquad \bm \eta_{0,j}=0.
\]
Then for $j>\lfloor \frac{N}{2}\rfloor$,
\[
    \bm \eta_{n+1,j}  = \bm x_{n+1}^{\otimes j} - \bm y_{n+1,j} = \sum_{k=0}^N L_{j,k}(\bm x_n^{\otimes k }-\bm y_{n,k})+\sum_{k=N+1}^{2j}L_{j,k}\bm x_n^{\otimes k}= \sum_{k=0}^N L_{j,k}\bm \eta_{n,k}+\sum_{k=N+1}^{2j}L_{j,k}\bm x_n^{\otimes k}.
\]
For $j\le \lfloor \frac{N}{2}\rfloor$, one has
\[ \bm \eta_{n+1,j}  =  \sum_{k=0}^N L_{j,k}\bm \eta_{n,k}. \]
Define $\delta\coloneqq R^{N+1}.$ We assume that  $\norm{\bm \eta_{n,j}}\le r_+^{j}\delta$ for every $j=1,\cdots,N$ as the induction hypothesis.
When $2j \le N$, then 
\[ \norm{\bm \eta_{n+1,j}}\le \sum_{k=1}^{2j}\norm{L_{j,k}}r_+^k\delta \le \delta \sum_{k=1}^{2j}m_{j,k}r_+^k \le \delta h(r_+)^j = \delta r_+^j. \]
For the case where $2j>N$,
\[ \norm{\bm \eta_{n+1,j}}\le \sum_{k=1}^{N}\norm{L_{j,k}}\norm{\bm \eta_{n,k}}+\sum_{k=N+1}^{2j}\norm{L_{j,k}}\norm{\bm x_n^{\otimes k}} \le \delta \sum_{k=1}^{2j}m_{j,k}r_+^k = \delta h(r_+)^j = \delta r_+^j.\]
By induction, we know for every $n\in\N_0$,
\[ \norm{\bm \eta_{n,j}}\le r_+^{j}\delta. \]
Take $j=1$, we obtain
\[ \norm{\bm \eta_{n,1}}\le r_+ R^{N+1} = \rho R^N.\]

\end{proof}

\section{Proof of \cref{lem: recovery_Carleman}} \label{pf: lem-recovery-Carleman}
\begin{proof}
    Fix an arbitrary $\bm y=\mat{\bm y_0, \cdots , \bm y_N}\in\V_N$, and let $\bm w$ be the unique solution of
    \[ (I-C^\dag)\bm w=E_{\mathrm{out}}\bm y. \]
    The block equations give
    \[ \bm w^{(0)}=\bm y, \quad \bm w^{(\ell)}=Q_\ell^\dag\bm w^{(\ell-1)},\quad 1\le\ell\le N.\]
Using \cref{eq:sparse-Q}, this becomes
\begin{equation}\label{eq:sparse-w-recursion}
\bm w_{p,\ell}
=
\sum_{\substack{i\in\{0,1,2\}\\p+i\le N}}
\left(
I_{\H}^{\otimes p}\otimes F_i
\otimes I_{\H}^{\otimes(\ell-1)}
\right)
\bm w_{p+i,\ell-1}.
\end{equation}

We claim that
\begin{equation}\label{eq:sparse-w-expansion}
\bm w_{p,\ell}
=
\sum_{\substack{i_1,\ldots,i_\ell\in\{0,1,2\}\\
p+i_1+\cdots+i_\ell\le N}}
\left(
I_{\H}^{\otimes p}\otimes
F_{i_1}\otimes\cdots\otimes F_{i_\ell}
\right)
\bm y_{p+i_1+\cdots+i_\ell}.
\end{equation}
We prove \cref{eq:sparse-w-expansion} by induction on $\ell$,
simultaneously for all $p\in\{0,\ldots,N\}$.
For $\ell=0$, the empty tensor product is the identity, and also $\bm w^{(0)} = \bm y$, thus \eqref{eq:sparse-w-expansion} holds for every $p$. Now suppose \eqref{eq:sparse-w-expansion} holds at layer $\ell-1$ as induction hypothesis for every $p$.
Substituting it into \eqref{eq:sparse-w-recursion} yields
\[
\begin{aligned}
\bm w_{p,\ell}
=&
\sum_{\substack{i_1\in\{0,1,2\}\\p+i_1\le N}}
\ \sum_{\substack{i_2,\ldots,i_\ell\in\{0,1,2\}\\
p+i_1+\cdots+i_\ell\le N}}
\\
&\quad
\left(
I_{\H}^{\otimes p}\otimes F_{i_1}
\otimes I_{\H}^{\otimes(\ell-1)}
\right)
\left(
I_{\H}^{\otimes(p+i_1)}
\otimes F_{i_2}\otimes\cdots\otimes F_{i_\ell}
\right)
\bm y_{p+i_1+\cdots+i_\ell}.
\end{aligned}
\]
By Kronecker product properties,
\[ \left(
I_{\H}^{\otimes p}\otimes F_{i_1}
\otimes I_{\H}^{\otimes(\ell-1)}
\right)
\left(
I_{\H}^{\otimes(p+i_1)}
\otimes F_{i_2}\otimes\cdots\otimes F_{i_\ell}
\right) =
I_{\H}^{\otimes p}\otimes
F_{i_1}\otimes F_{i_2} \otimes \cdots\otimes F_{i_\ell}.
\]
Furthermore, all $i_t$ are nonnegative, so
\[
p+i_1+\cdots+i_\ell\le N
\quad\Longrightarrow\quad
p+i_1\le N.
\]
The outer cutoff is therefore redundant, and the resulting
sum is exactly \eqref{eq:sparse-w-expansion}.
This proves the claim by induction.

Taking $p=0$ and $\ell=j\ge1$, and grouping terms by
$k=i_1+\cdots+i_j$, gives
\[
\bm w_{0,j} = \sum_{k=0}^N
\left( \sum_{\substack{i_1,\ldots,i_j\in\{0,1,2\} \\
i_1+\cdots+i_j=k}} F_{i_1}\otimes\cdots\otimes F_{i_j} \right)\bm y_k=\sum_{k=0}^NL_{j,k}\bm y_k.
\]
Also, $w_{0,0}=\bm y_0$. Thus
\[
\Ein^\dag\bm w=L_N\bm y.
\]
Note that $\bm y$ is arbitrary, so $L_N = \Ein^\dag (I-C^\dag)^{-1}\Eout$.
\end{proof}

\section{Proof of \cref{lem:sparsity-D}} \label{pf: lem-sparsity-D}
\begin{proof}
Recall that $F_0=\bm f$, $F_1=A$, and $F_2=B$.
The nonzero blocks of $C$ are
\[
C_{(p+i,\ell-1),(p,\ell)}
=
I_{\H}^{\otimes p}\otimes F_i^\dag
\otimes I_{\H}^{\otimes(\ell-1)},
\]
where $i\in\{0,1,2\}$, $1\le\ell\le N$, and
$0\le p\le p+i\le N$.

Taking the adjoint exchanges row and column sparsity.
Tensoring with identity matrices preserves the maximum
number of nonzero entries in each row and column.
Thus the blocks corresponding to $i=0,1,2$ have row and
column sparsity bounded by $s_f,s_A,s_B$, respectively.
For $i=0$, this follows because $\bm f^\dag$ has at most $s_f$
nonzero entries in total.

Fix a source sector $(p,\ell)$. For each $i\in\{0,1,2\}$,
there is at most one destination sector,
\[
(p+i,\ell-1).
\]
Hence each block column of $C$ contains at most one block
of each type. Every scalar column of $C$ therefore has
at most $s_f+s_A+s_B$ nonzero entries.

Similarly, for a fixed destination sector $(p',\ell')$
and a fixed $i$, the only possible source sector is
\[
(p'-i,\ell'+1).
\]
Thus each block row of $C$ contains at most one block
of each type, and every scalar row of $C$ has at most
$s_f+s_A+s_B$ nonzero entries.

The same row and column bounds hold for $C^\dag$.
Since $C^\dag$ changes the layer index, its diagonal is zero.
Adding $I_{\K}$ therefore contributes exactly one additional
nonzero entry to each row and column, proving
\cref{eq:D-sparsity}.
\end{proof}

\section{Proof of \cref{lem:hs-uniform}}
\label{pf: lem-hs-uniform}

\begin{proof}
For $i\in\{0,1,2\}$, let $\widetilde C_i$ contain precisely
the blocks of $\widetilde C$ involving $\widetilde F_i^\dag$.
By \cref{eq:sparse-Q,eq:sparse-C}, these blocks are
\[
    (\widetilde C_i)_{(p+i,\ell-1),(p,\ell)}
    =
    I_{\H}^{\otimes p}\otimes\widetilde F_i^\dag
    \otimes I_{\H}^{\otimes(\ell-1)},
\]
where $1\le\ell\le N$ and $0\le p\le p+i\le N$.
For each fixed $i$, distinct source sectors map to distinct
destination sectors. Orthogonality of these sectors and
multiplicativity of the operator norm under tensor products give
\[
    \norm{\widetilde C_i\bm v}^2
    \le
    \norm{\widetilde F_i}^2
    \sum_{\substack{1\le\ell\le N\\0\le p\le p+i\le N}}
        \norm{\bm v_{p,\ell}}^2
    \le
    \norm{\widetilde F_i}^2\norm{\bm v}^2.
\]
Consequently, using \cref{eq:hs-scaled-map},
\[
    \norm{\widetilde C}
    \le\sum_{i=0}^2\norm{\widetilde C_i}
    \le\sum_{i=0}^2\norm{\widetilde F_i}
    \le\frac cr+\alpha+\beta r=q<1.
\]
Since $\widetilde D=I_{\K}-\widetilde C^\dag$, we obtain
\[
    \norm{\widetilde D}\le1+q.
\]
Moreover, $\widetilde C^{N+1}=0$, so
\[
    \widetilde D^{-1}
    =\sum_{\ell=0}^N(\widetilde C^\dag)^\ell,
    \qquad
    \norm{\widetilde D^{-1}}
    \le\sum_{\ell=0}^Nq^\ell
    \le\frac1{1-q}=\delta.
\]
Therefore,
\[
    \kappa(\widetilde D)
    =\norm{\widetilde D}\norm{\widetilde D^{-1}}
    \le(1+q)\delta.
\]

Next, $\Eout$ is an isometry into layer zero, while
$\widetilde C^\dag$ increases the layer index by one.
Thus, for any $\bm v\in\V_N$, the summands in
\[
    J\bm v
    =\sum_{\ell=0}^N
      (\widetilde C^\dag)^\ell\Eout\bm v
\]
belong to mutually orthogonal layers. Hence
\[
    \norm{J\bm v}^2
    =\sum_{\ell=0}^N
      \norm{(\widetilde C^\dag)^\ell\Eout\bm v}^2
    \le\sum_{\ell=0}^Nq^{2\ell}\norm{\bm v}^2
    \le\frac{\norm{\bm v}^2}{1-q^2}.
\]
Taking the supremum over unit vectors proves
$\norm J\le (1-q^2)^{-1/2}$, completing \cref{eq:hs-D-bounds}.

To prove \cref{eq:hs-trajectory-bounds}, define
\[
    \widetilde h(s)
    :=\frac cr+\alpha s+\beta r s^2.
\]
Since $r_-\le\rho<r_+$ and $a=\rho/r$,
\[
    \widetilde h(a)-a
    =\frac{c+\alpha\rho+\beta\rho^2-\rho}{r}
    =\frac{\beta}{r}(\rho-r_-)(\rho-r_+)
    \le0.
\]
For $j\ge1$, write
\[
    \widetilde h(s)^j
    =\sum_{k=0}^{2j}\widetilde m_{j,k}s^k,
    \qquad \widetilde m_{j,k}\ge0.
\]
The tensor-product formula \cref{eq:carleman-blocks},
applied to the rescaled coefficients, gives
\[
    \norm{\widetilde L_{j,k}}
    \le\widetilde m_{j,k}.
\]
Here $\widetilde L_{j,k}$ denotes the corresponding
rescaled Carleman block.

We prove $\norm{\widetilde{\bm y}_{t,j}}\le a^j$
by induction on $t$, simultaneously for $1\le j\le N$.
The initial condition in \cref{eq:hs-recovery} gives
\[
    \norm{\widetilde{\bm y}_{0,j}}
    =\left(\frac{\norm{\bm x_0}}r\right)^j
    \le a^j,
\]
and the degree-zero coordinate remains equal to one.
Assuming the bounds hold at time $t$, we obtain
\[
\norm{\widetilde{\bm y}_{t+1,j}}\le \sum_{k=0}^{\min\{N,2j\}}\norm{\widetilde L_{j,k}}\norm{\widetilde{\bm y}_{t,k}}\le\sum_{k=0}^{\min\{N,2j\}}\widetilde m_{j,k}a^k\le \widetilde h(a)^j
    \le a^j.
\]
The $k=0$ term uses
$\widetilde{\bm y}_{t,0}=1$ and $a^0=1$.
Thus the blockwise bound holds for every $t\ge0$.
Orthogonality of the degree blocks now yields
\[
    \norm{\widetilde{\bm y}_t}^2
    =\sum_{j=0}^N
        \norm{\widetilde{\bm y}_{t,j}}^2
    \le\sum_{j=0}^Na^{2j}
    \le\frac1{1-a^2}.
\]
Combining this with $\norm J\le g$ gives
\[
    \norm{J\widetilde{\bm y}_t}
    \le g\norm{\widetilde{\bm y}_t}
    \le\frac{g}{\sqrt{1-a^2}}=M_r,
\]
which proves \cref{eq:hs-trajectory-bounds}.

Finally, we establish the uniform bound from \cref{eq: L_N_power_bound}.
Since $\widetilde h(1)=q<1$, any vector
$\bm v=(\bm v_0,\ldots,\bm v_N)\in\V_N$ satisfying
$\norm{\bm v_j}\le b$ for all $j$ obeys
\[
    \norm{(\widetilde L_N\bm v)_j}
    \le b\sum_{k=0}^{\min\{N,2j\}}
        \widetilde m_{j,k}
    \le bq^j
    \le b,
    \qquad 1\le j\le N.
\]
The degree-zero block is unchanged and also has norm at most $b$.
Starting with $b=\norm{\bm v}$ and iterating this argument gives
\[
    \norm{(\widetilde L_N^k\bm v)_j}
    \le q^j\norm{\bm v},
    \qquad k\ge1,\quad 0\le j\le N.
\]
Therefore,
\[
    \norm{\widetilde L_N^k\bm v}^2
    \le\sum_{j=0}^Nq^{2j}\norm{\bm v}^2
    \le g^2\norm{\bm v}^2.
\]
Thus $\norm{\widetilde L_N^k}\le g$ for every $k\ge1$, and 
the same bound holds for $k=0$ because $g\ge1$.
All bounds are uniform in $N$.
\end{proof}

\section{Proof of \cref{lem:hs-history}} \label{pf: lem-hs-history}
\begin{proof}
Write $M=\dim\K=m^2$. Substituting
$\widetilde D=I_\K-\widetilde C^\dag$ gives
\[
    \widetilde{\mathcal A}_n=I_{nM}-Z,
    \qquad Z=I_n\otimes\widetilde C^\dag+S_n\otimes P.
\]
The first summand of $Z$ has range in positive auxiliary layers, and the second has range in layer zero because $P=\Eout\Ein^\dag$.
These ranges are orthogonal. Since $\norm{S_n}\le1$ and $\norm P=1$,
\[
    \norm{Z\bm v}^2
    =\norm{(I_n\otimes\widetilde C^\dag)\bm v}^2
     +\norm{(S_n\otimes P)\bm v}^2
    \le(q^2+1)\norm{\bm v}^2.
\]
Thus $\norm Z\le\sqrt{1+q^2}$. Applying the triangle inequality to
$\widetilde{\mathcal A}_n=I_{nM}-Z$ now gives
\[
    \norm{\widetilde{\mathcal A}_n}
    \le\norm{I_{nM}}+\norm Z\le1+\sqrt{1+q^2}.
\]

Now we bound the inverse norm, $\norm{\widetilde{\mathcal{A}}_n^{-1}}$. Set $T=\widetilde D^{-1}P=J\Ein^\dag$. Then
\[
    \widetilde{\mathcal A}_n
    =(I_n\otimes\widetilde D)(I_{nM}-S_n\otimes T).
\]
Since $S_n^n=0$, the second factor has inverse
$\sum_{k=0}^{n-1}S_n^k\otimes T^k$. Reversing the factors therefore yields
\begin{equation}\label{eq:hs-inverse-series}
    \widetilde{\mathcal A}_n^{-1}
    =\sum_{k=0}^{n-1}S_n^k\otimes T^k\widetilde D^{-1}.
\end{equation}
The recovery identity says $\Ein^\dag J=\widetilde L_N$.
Starting with $T=J\Ein^\dag$, induction gives
\[
    T^{k+1}
    =(J\widetilde L_N^{k-1}\Ein^\dag)(J\Ein^\dag)
    =J\widetilde L_N^k\Ein^\dag,
    \qquad
    T^k=J\widetilde L_N^{k-1}\Ein^\dag, \quad\text{for }k\ge1.
\]
Using \cref{lem:hs-uniform} and $\norm{\Ein^\dag}=1$, we obtain
\[
    \norm{T^k\widetilde D^{-1}}
    \le\norm J\norm{\widetilde L_N^{k-1}}
        \norm{\Ein^\dag}\norm{\widetilde D^{-1}}
    \le g^2\delta=K_r,\quad\text{for }k\ge1.
\]
For $k=0$ the bound is $\delta$. Taking norms in
\eqref{eq:hs-inverse-series} and using $\norm{S_n^k}\le1$ gives
\[
    \norm{\widetilde{\mathcal A}_n^{-1}}
    \le\delta+(n-1)K_r
    =\delta\bigl(1+(n-1)g^2\bigr)\le nK_r.
\]
Multiplying by the preceding norm bound proves the upper bound on $\kappa(\widetilde{\mathcal A}_n)$.

Let $\bm e$ be the unit vector in the scalar sector $(p,\ell)=(0,0)$,
and solve
$\widetilde{\mathcal A}_n\bm v=(\bm e,\ldots,\bm e)^T$.
The scalar row of $\widetilde D$ is the identity row, and
$(P\bm v_t)_{0,0}=(\bm v_t)_{0,0}$. The scalar equations are therefore
\[
    (\bm v_0)_{0,0}=1,\qquad
    (\bm v_t)_{0,0}-(\bm v_{t-1})_{0,0}=1,
\]
so $(\bm v_t)_{0,0}=t+1$ by induction. The squared norm of the
right-hand side is $n$, whereas
$\norm{\bm v}^2\ge\sum_{t=0}^{n-1}(t+1)^2$. Consequently,
\[
    \norm{\widetilde{\mathcal A}_n^{-1}}
    \ge\sqrt{\frac1n\sum_{t=0}^{n-1}(t+1)^2}
    =\sqrt{\frac{(n+1)(2n+1)}6}\ge\frac n{\sqrt3}.
\]
Every diagonal entry of $\widetilde{\mathcal A}_n$ is one, so
$\norm{\widetilde{\mathcal A}_n}\ge1$.
The definition of the condition number proves the lower bound.
\end{proof}

\section{Proof of \cref{lem:hs-implementation}} \label{pf: lem-hs-implementation}
\begin{proof}
The diagonal and subdiagonal parts of $\mathcal B_n$ have norms at most $1+q$ and $1$, respectively.
Writing $T=\widetilde D^{-1}P$, forward substitution shows that
every nonzero block of $\mathcal B_n^{-1}$ is either
$T^j\widetilde D^{-1}$ or $I_\K$.
The estimates in the proof of \cref{lem:hs-history} bound these
blocks by $K_r$.
Their block row and column sums are at most $2nK_r$,
proving \eqref{eq:hs-padded-bounds}.

By \cref{lem:sparsity-D} and the identity-copy construction,
$\mathcal B_n$ has row and column sparsity at most $s$.
Its entries have magnitude at most one, and its sparse-access
oracles require $O(1)$ coefficient queries.
The sparse block-encoding construction of
\cite{gilyen2019quantum} therefore encodes $\mathcal B_n/s$.
The Hermitian dilation
\begin{equation}\label{eq:hs-dilation}
    H_{\mathcal B}
    =\frac1s
      \begin{pmatrix}
        0&\mathcal B_n\\
        \mathcal B_n^\dag&0
      \end{pmatrix}
\end{equation}
satisfies
\[
    \norm{H_{\mathcal B}}\le1,
    \qquad
    \norm{H_{\mathcal B}^{-1}}\le2snK_r,
    \qquad
    H_{\mathcal B}^{-1}\binom{\bm b}{0}
      =\binom{0}{s\bm\Omega}.
\]
The claimed solver complexity follows from
\cite[Theorem~11]{costa2022optimal}, with measurements deferred
to coherent success flags.
\end{proof}

\section{Proof of \cref{thm:hs-quantum}} \label{pf: thm-hs-quantum}

\begin{proof}
By \cref{thm: uniform} and the choice of $N$,
\[
    \norm{\bm x_n-\bm y_{n,1}}
       \le\rho R^N\le\frac{\varepsilon\xi}{8},
    \qquad
    h_n=\frac{\norm{\bm y_{n,1}}}{r}
       \ge\frac{\xi}{2r}.
\]
Thus \eqref{eq:hs-success} gives
$p_{\mathrm{out}}\ge1/(2\Gamma^2)$.
The inequality
\[
    \left\|
      \frac{\bm u}{\norm{\bm u}}
      -\frac{\bm v}{\norm{\bm v}}
    \right\|
    \le\frac{2\norm{\bm u-\bm v}}{\norm{\bm u}}
\]
bounds the normalized truncation error by $\varepsilon/4$.

Apply \cref{lem:hs-implementation} with
$\eta=\varepsilon/(32\Gamma)$.
Projection and normalization introduce error at most
\[
    \frac{2\eta}{\sqrt{p_{\mathrm{out}}}}
    \le2\sqrt2\,\Gamma\eta<\frac{\varepsilon}{4}.
\]
Moreover, the approximate projection succeeds with probability
at least
$(\sqrt{p_{\mathrm{out}}}-\eta)^2
 \ge p_{\mathrm{out}}/4$.
Including the solver flag, the joint success probability is
therefore $\Omega(\Gamma^{-2})$.
Fixed-point amplitude amplification
\cite{yoder2014fixedpoint} raises it to at least $2/3$
using $O(\Gamma)$ coherent solver calls, proving
\eqref{eq:hs-query-bound}.

By \cref{eq:hs-projected-state} and \cref{eq:hs-physical-recovery},
the ideal projected physical state is
$\bm y_{n,1}/\norm{\bm y_{n,1}}$.
Trace distance is bounded by the corresponding vector error
and cannot increase under discarding auxiliary registers.
The two error contributions therefore sum to less than
$\varepsilon/2$, proving the claim.
\end{proof}

\end{document}